\documentclass[pra,showpacs,superscriptaddress,amssymb,twocolumn,longbibliography]{revtex4-1}
\usepackage{graphicx}
\usepackage{amsmath}
\usepackage{amsthm}
\usepackage{amssymb}
\usepackage[usenames]{color}
\usepackage{hyperref}
\hypersetup{
    colorlinks=true,       
    linkcolor=cyan,          
    citecolor=magenta,        
    filecolor=magenta,      
    urlcolor=cyan,           
    runcolor=cyan
}

\usepackage{bm}
\usepackage{threeparttable}
\usepackage{subfigure}
\usepackage{upgreek }
\usepackage{marginnote}
\def\wich#1{\mathinner{\langle{#1}\rangle}}
\newcommand{\beq}{\begin{equation}}
\newcommand{\eeq}{\end{equation}}
\newcommand{\beqnn}{\begin{equation*}}
\newcommand{\eeqnn}{\end{equation*}}
\newcommand{\bea}{\begin{eqnarray}}
\newcommand{\eea}{\end{eqnarray}}
\newcommand{\beann}{\begin{eqnarray*}}
\newcommand{\eeann}{\end{eqnarray*}}
\newcommand{\bes} {\begin{subequations}}
\newcommand{\ees} {\end{subequations}}

\newcommand{\ket}[1]{ | #1\rangle}
\newcommand{\bra}[1]{\langle #1 | }

\newcommand{\mO}{\mathcal{O}}

\newcommand{\ident}{\openone}
\newcommand{\Tr}{\mathrm{Tr}}
\newcommand{\eps}{\varepsilon}
\newcommand{\abs}[1]{\ensuremath{\left| #1 \right|}}

\def\Ket#1{\left|#1\right>}

\newcommand{\ignore}[1]{}
\newtheorem{theorem}{Theorem}
\newtheorem{lemma}{Lemma}

\begin{document}
\title{Tunneling and speedup in quantum optimization for permutation-symmetric problems}
\author{Siddharth Muthukrishnan}
\affiliation{Department of Physics and Astronomy, University of Southern California, Los Angeles, California 90089, USA}
\affiliation{Center for Quantum Information Science \& Technology, University of Southern California, Los Angeles, California 90089, USA}

\author{Tameem Albash}
\affiliation{Department of Physics and Astronomy, University of Southern California, Los Angeles, California 90089, USA}
\affiliation{Center for Quantum Information Science \& Technology, University of Southern California, Los Angeles, California 90089, USA}
\affiliation{Information Sciences Institute, University of Southern California, Marina del Rey, CA 90292}

\author{Daniel A. Lidar}
\affiliation{Department of Physics and Astronomy, University of Southern California, Los Angeles, California 90089, USA}
\affiliation{Center for Quantum Information Science \& Technology, University of Southern California, Los Angeles, California 90089, USA}
\affiliation{Department of Electrical Engineering, University of Southern California, Los Angeles, California 90089, USA}
\affiliation{Department of Chemistry, University of Southern California, Los Angeles, California 90089, USA}

\begin{abstract}
Tunneling is often claimed to be the key mechanism underlying possible speedups in quantum optimization via quantum annealing (QA), especially for problems featuring a cost function with tall and thin barriers. We present and analyze several counterexamples from the class of perturbed Hamming-weight optimization problems with qubit permutation symmetry.  We first show that, for these problems, the adiabatic dynamics that make tunneling possible should be understood not in terms of the cost function but rather the semi-classical potential arising from the spin-coherent path integral formalism. We then provide an example where the shape of the barrier in the final cost function is short and wide, which might suggest no quantum advantage for QA, yet where tunneling renders QA superior to simulated annealing in the adiabatic regime. However, the adiabatic dynamics turn out not be optimal. Instead, an evolution involving a sequence of diabatic transitions through many avoided level-crossings, involving no tunneling, is optimal and outperforms adiabatic QA. We show that this phenomenon of speedup by diabatic transitions is not unique to this example, and we provide an example where it provides an exponential speedup over adiabatic QA.  In yet another twist, we show that a classical algorithm, spin vector dynamics, is at least as efficient as diabatic QA. Finally, in a different example with a convex cost function, the diabatic transitions result in a speedup relative to both adiabatic QA with tunneling and classical spin vector dynamics.

\end{abstract} 
\maketitle
%
\section{Introduction}
%
The possibility of a quantum speedup for finding the solution of classical optimization problems is tantalizing, as a quantum advantage for this class of problems would provide a wealth of new applications for quantum computing. The goal of many optimization problems can be formulated as finding an $n$-bit string $x_{\mathrm{opt}}$ that minimizes a given cost function $f(x)$, which can be interpreted as the energy of a classical 
Ising spin system whose ground state is $x_{\mathrm{opt}}$.  Finding the ground state of such systems can be hard if, e.g., the system is strongly frustrated, resulting in a complex energy landscape that cannot be efficiently explored with any known algorithm due to the presence of many local minima \cite{Nishimori-book}. 
This can occur, e.g., in classical simulated annealing (SA) \cite{kirkpatrick_optimization_1983}, 
when the system's state is trapped in a local minimum.  

Thermal hopping and quantum tunneling provide two starkly different mechanisms for solving optimization problems, and finding optimization problems that favor the latter continues to be an open theoretical question \cite{EPJ-ST:2015,Heim:2014jf}. 
It is often stated that quantum annealing
(QA) \cite{PhysRevB.39.11828,finnila_quantum_1994,kadowaki_quantum_1998,farhi_quantum_2001,RevModPhys.80.1061}  uses tunneling 
instead of thermal excitations to escape from local
minima, which can be advantageous in systems
with tall but thin barriers that are
easier to tunnel through than to thermally climb
over \cite{Heim:2014jf,RevModPhys.80.1061,Suzuki-book}. 
It is with this potential tunneling-induced advantage over classical annealing that QA and the quantum adiabatic algorithm \cite{farhi_quantum_2000} were proposed. Our goal in this work is to address the question of the role played by tunneling in providing a quantum speedup, and to elucidate it by studying a number of illustrative examples. We shall demonstrate that the role of tunneling is significantly more subtle than what might be expected on the basis of the ``tall and thin barrier'' picture.

In order to make progress on this question, the potential with respect to which tunneling  occurs must be clearly specified. Tunneling is defined with respect to a semi-classical potential which delineates classically allowed and forbidden regions. In QA, one typically initializes the system in the known ground state of a simple Hamiltonian and evolves the system towards a Hamiltonian representing the final cost function. 
We shall argue that when one takes a natural semi-classical limit, the semi-classical potential does not become the final cost-function.  Instead one obtains a potential appearing in the action of the spin-coherent path-integral representation of the quantum dynamics. This potential, which here we call the spin-coherent potential, has been used profitably before~\cite{Farhi-spike-problem,FarhiAQC:02,Schaller:2007uq,Boixo:2014yu}. We provide comprehensive evidence that multi-spin tunneling can be understood with respect to this spin-coherent potential. 

We analyze the spin-coherent potential for several examples from a well-known class of problems known as perturbed Hamming weight oracle (PHWO) problems. These are problems for which instances can be generated where QA either has an advantage over classical random search algorithms with local updates, such as SA \cite{Farhi-spike-problem,Reichardt:2004}, or has no advantage \cite{vanDam:01,Reichardt:2004}. Moreover, because PHWO problems exhibit qubit permutation symmetry, their quantum evolutions are easily classically simulatable, and furthermore, their spin-coherent potential is one-dimensional.  Tunneling becomes clear and explicit for these problems when using the spin-coherent potential. 

We focus on a particular PHWO problem that has a plateau in the final cost function (henceforth,``the Fixed Plateau"). This problem offers a counter-example to two commonly held views: (1) QA has an advantage, due to tunneling, over SA only on problems where the barrier in the final cost function is tall and thin; (2) tunneling is necessary for a quantum speedup in QA.  We refute the first statement by showing that for the Fixed Plateau, which is a short and wide cost function, QA significantly outperforms SA by using tunneling. Indeed, we find numerically that adiabatic QA (AQA) needs a time of $\mathcal{O}(n^{0.5})$ to find the ground state, where $n$ is the number of spins or qubits. Moreover, using the spin-coherent potential, we observe the presence of tunneling during the quantum anneal. On the other hand, we prove that single-spin update SA takes a time of $\mathcal{O}(n^\mathrm{plateau\ width})$. Thus, we have essentially an arbitrary polynomial tunneling speedup of QA over SA on a cost-function that is not tall and thin. We remark that the result about SA's performance is also a rigorous proof of a result due to Reichardt~\cite{Reichardt:2004} that classical local search algorithms will fail on a certain class of PHWO problems and is of independent interest. 

We refute the second statement by showing that, for the Fixed Plateau, it is actually optimal to run QA diabatically (henceforth, DQA for diabatic quantum annealing).  The system leaves the ground state, only to return through a sequence of diabatic transitions associated with avoided-level crossings.  In this regime, the runtime for QA is $\mathcal{O}(1)$. Moreover, in this regime, we do not observe any of the standard signatures of tunneling. We show that this feature --- that the optimal evolution time $t_f$ for QA is far from being adiabatic --- is present in a few other PHWO problems and that this optimal evolution involves no multi-qubit tunneling.  

Given that the optimal evolution involves no tunneling, we are inspired to investigate a classical algorithm, spin vector dynamics (SVD), which can be interpreted as a semi-classical limit of the quantum evolution with a product-state approximation. We observe that SVD evolves in an almost identical manner to DQA, and is able to recover the speedup seen by DQA.  Thus, in these problems, we show that what may be suspected to be a highly quantum-coherent process---diabatic transitions---can be mimicked by a quantum-inspired classical algorithm.

The structure of this paper is as follows. In Sec.~\ref{sec:PHWOintro}, we list the PHWO problems we study. In Sec.~\ref{sec:SCtunnel}, we use these problems to present evidence that tunneling can be understood with respect to the spin-coherent potential. In Sec.~\ref{sec:FixedPlateau}, we focus on the Fixed Plateau PHWO problem, and exhaustively analyze the performance of various algorithms for this problem. In particular we numerically characterize AQA (Sec.~\ref{subsec:adDyn}), provide a rigorous proof of SA's performance (Sec.~\ref{subsec:SArandom}), and numerically analyze DQA (Sec.~\ref{subsec:diabaticQA}), SVD (Sec.~\ref{subsec:SVDplat}), and a quantum Monte Carlo algorithm (Sec.~\ref{subsec:SQAplat}). We conclude in Sec.~\ref{sec:discuss} by discussing the implications of our work and possible directions for future work. Additional background information and technical details can be found in the Appendix.

\section{Perturbed Hamming weight optimization problems and the examples studied}\label{sec:PHWOintro}
%

The cost function of a PHWO problem is defined as,

\beq 
\label{eq:pertham}
f(x) = \begin{cases} \abs{x} + p(\abs{x}) & l<\abs{x}<u,\\
\abs{x} & \text{elsewhere} \end{cases} \ , 
\eeq
where $\abs{x}$ denotes the Hamming weight of the bit string $x\in\{0,1\}^n$. For SA, this is the cost-function. For QA, this will be the final Hamiltonian. More precisely, we define QA as the closed-system quantum evolution governed by the time-dependent Hamiltonian,
\beq 
\label{eqt:QuantumH}
H(s) =\frac{1}{2}\left(1-s \right)  \sum_i \left(\ident - \sigma_i^x \right)+ s \sum_{x} f(x) \ket{x} \bra{x} \ ,
\eeq
where we have chosen the standard transverse field ``driver'' Hamiltonian $H(0)$ that assumes no prior knowledge of the form of $f(x)$, and a linear interpolating schedule, with $s\equiv t/t_f$ being the dimensionless time parameter. The initial state is the ground state of $H(0)$.

Below, we list several of PHWO examples that we study in greater detail. We refer to the case with $p=0$ as the \textbf{Plain Hamming Weight} problem.

\begin{enumerate}
\item \textbf{Fixed Plateau:} 
\beq 
\label{eqt:plateau}
f(x) = \begin{cases}
u -1, &  l<\abs{x}<u,\\
\abs{x}, & \text{otherwise}
\end{cases} \ .
\eeq
Clearly, this forms a plateau in Hamming weight space. We take $u,l=\mathcal{O}(1)$.  Since the location of the plateau does not change with $n$, we refer to it as ``fixed."  An instance of this cost function with $l=3$ and $u=8$ is illustrated in Fig.~\ref{fig:Plateau}. By numerical diagonalization we find that QA has a constant gap for this cost-function.

\begin{figure}[t] 
   \centering
   \includegraphics[width=0.65\columnwidth]{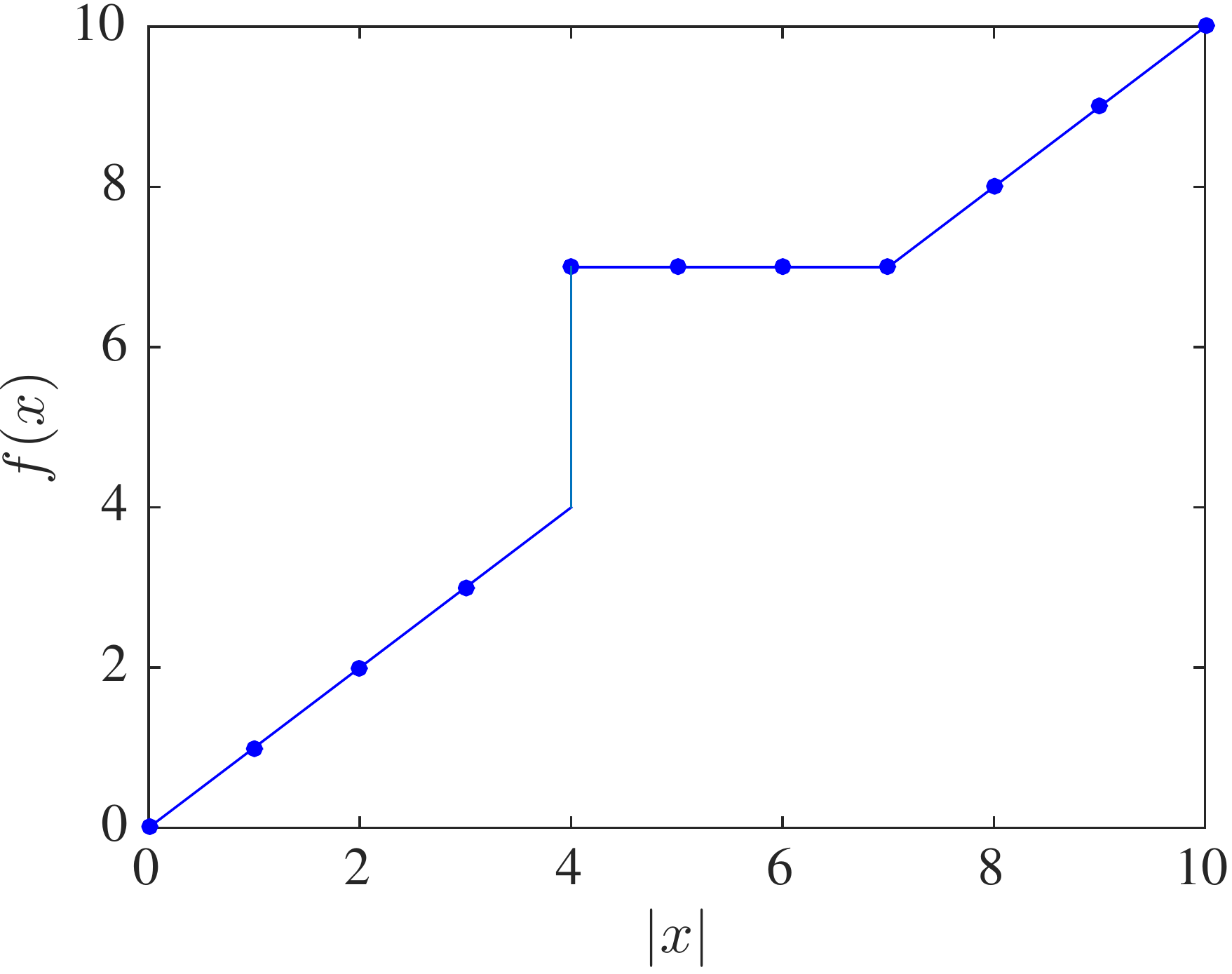} 
   \caption{$l=3,u=8$}
   \label{fig:Plateau}
\end{figure}

\item \textbf{Reichardt:} 
\beq 
\label{eqt:Reichardt}
f(x) = \begin{cases}
\abs{x}+h(n), &  l(n)<\abs{x}<u(n)\\
\abs{x} & \text{otherwise}
\end{cases} \ ,
\eeq
with $h\frac{u-l}{\sqrt{l}}= o(1)$. For this case, Reichardt~\cite{Reichardt:2004} proved a constant lower bound on the minimum spectral gap during the quantum anneal. In Appendix~\ref{app:review} we provide a pedagogical review of this proof and fill in some details not explicitly provided in the original proof.

\item \textbf{Moving Plateau:}
\beq 
\label{eqt:movingplateau}
f(x) = \begin{cases}
u -1, &  l(n)<\abs{x}<u(n)\\
\abs{x} & \text{otherwise}
\end{cases} \ ,
\eeq
with $l(n) = n/4$, and $u(n)=\mathcal{O}(1)$. This is termed ``moving" since the location of the plateau changes with $n$. Note that this is a special case from the Reichardt class.

\item \textbf{Grover:} 
\beq 
\label{eqt:Grover}
f(x) = \begin{cases}
n, &  \abs{x}\geq 1,\\
0, & \abs{x}=0.
\end{cases} \ 
\eeq
This is a minor modification of the standard Grover problem: the marked state is the all-zeros string with energy $0$, and the energy of all the other states is $n$. Scaling the energy by $n$ keeps the maximum energy of all the PHWO problems we consider comparable. 

\item \textbf{Spike:}
\beq 
\label{eqt:spike}
f(x) = \begin{cases}
n, &  \abs{x}=n/4,\\
\abs{x}, & \text{otherwise}
\end{cases} \ .
\eeq
This was studied by Farhi \textit{et al}. in~\cite{Farhi-spike-problem}, where it was argued that the quantum minimum gap  scales as $\mathcal{O}(n^{-1/2})$ and that SA will take exponential time to find the ground state. However, we show below (Fig.~\ref{fig:SpikeMovingVanDam_TTSopt}) that SVD is more efficient than QA for this problem.

\item \textbf{Precipice:}
\beq
\label{eqt:precipice}
f(x) = \begin{cases}
-1, &  \abs{x}=n,\\
\abs{x}, & \text{otherwise}
\end{cases} \ .
\eeq
This was studied by van Dam \textit{et al}. in~\cite{vanDam:01}, where it was proved that the quantum minimum gap for this problem scales as $\mathcal{O}(2^{-n/2})$.

\item \textbf{$\alpha$-Rectangle:} 
\beq
\label{eqt:alpharect}
f(x) = \begin{cases}
\abs{x} + n^\alpha, &   \frac{n}{4} - \frac{1}{2} c n^\alpha  <\abs{x}<  \frac{n}{4} + \frac{1}{2} c n^\alpha  ,\\
\abs{x}, & \text{otherwise}
\end{cases} \ .
\eeq
We call this an $\alpha$-Rectangle because the width of the perturbation ($c n^\alpha$) is $c$ times the height.  This was studied in~\cite{Brady:2015rc}, where evidence for the following conjecture for the scaling of the quantum minimum gap  $g_{\mathrm{min}}$ was presented,
\beq \label{eqt:VanDam}
g_\mathrm{min} = \begin{cases}
\mathrm{constant},  & \alpha<\frac{1}{4}, \\
1 / \mathrm{poly}(n),  & \frac{1}{4}<\alpha<\frac{1}{3}, \\
1 /  \exp(n), & \alpha>\frac{1}{3}.
\end{cases}
\eeq
Note that $\alpha<1/4$ is a member of the Reichardt class, and thus the constant lower-bound on the minimum gap is a theorem, and not a conjecture. We restrict ourselves to the case of $c=1$.

\end{enumerate}

We remark that all the problems listed above are representative members of a large family of problems: if the input bit-string to any of the above problems is transformed by an $\mathrm{XOR}$ mask, then all of our analysis below will hold. For QA, the $\mathrm{XOR}$ mask can be represented as a unitary transformation: $\bigotimes_{i=1}^n (\sigma_i^x)^{a_i}$, with $a \in \{0,1\}^n$ being the mask string. As this unitary commutes with the QA Hamiltonian at all times, none of our subsequent analysis is affected. Similar arguments go through for SA and all the other algorithms that we consider.

We note that PHWO problems are strictly toy problems since these problems are typically represented by highly non-local Hamiltonians (see Appendix~\ref{app:locality}) and thus are not physically implementable, in the same sense that the adiabatic Grover search problem is unphysical \cite{Roland:2002ul,RPL:10}. Nevertheless, these problems provide us with important insights into the mechanisms behind a quantum speed-up, or lack thereof.

%
\section{The semi-classical potential and tunneling }\label{sec:SCtunnel}
%
In order to study tunneling, we need a potential arising from a semi-classical limit, which defines classically allowed and forbidden regions. One approach to writing a semi-classical potential for quantum Hamiltonians is to use the spin-coherent path-integral formalism \cite{klauder1979path}. This semi-classical potential has been used profitably in various QA studies, e.g., Refs.~\cite{Farhi-spike-problem,Schaller:2007uq,FarhiAQC:02,Boixo:2014yu}, and we extend its applications here. For the quantum evolution, since the initial state [the ground state of $H(0)$] is symmetric under permutations of qubits and the unitary dynamics preserves this symmetry (it is a symmetry of $H(s)$ for all $s$), we can consistently restrict ourselves to spin-1/2 symmetric coherent states $\ket{\theta,\phi} $:
\beq
\ket{\theta,\phi} = \bigotimes_{i=1}^n \left[ \cos\left( \frac{\theta}{2} \right) \ket{0}_i + \sin \left( \frac{\theta}{2} \right) e^{i \varphi} \ket{1}_i \right] \ .
\eeq
The spin-coherent potential is then given by:
\beq \label{eq:vsc}
V_{{\mathrm{SC}}}(\theta,\phi,s) = \bra{\theta,\phi} H(s) \ket{\theta,\phi} \ .
\eeq
We show that for all the examples defined above except the Reichardt class (we address this below), this potential captures important features of the quantum Hamiltonian [Eq.~\eqref{eqt:QuantumH}] and reveals the presence of tunneling. Specifically:

\begin{enumerate}

\item {\emph{The spin-coherent potential displays a degenerate double-well almost exactly at the point of the minimum gap.}} In Fig.~\ref{fig:Veff} we plot, for the Fixed Plateau the potential near the minimum gap. The potential transitions from having a single minimum on the right to a single minimum on the left. In between, it becomes degenerate and displays a degenerate double well. Since the instantaneous ground state corresponds to the position of the global minimum, which exhibits a discontinuity, the degeneracy point is where tunneling should be most helpful. In Fig.~\ref{fig:MinGapsWellLocs}, we show that the location of the minimum gap of the quantum evolution is very close to the location of the degenerate double-well in the spin-coherent potential. 

\item {\emph{The ground state predicted by the spin-coherent potential is a good approximation to the quantum ground state except near the degeneracy point.}} As expected from a potential that arises in a semi-classical limit, the ground state predicted by the spin-coherent potential (i.e., the spin-coherent state corresponding to the instantaneous global minimum in $V_{\mathrm{SC}}$) agrees well with the quantum ground state, except where tunneling  is important. In particular, delocalization when the spin-coherent potential is a degenerate double-well (or is close to being one) should imply that approximating the ground state with a wavefunction localized in one of the wells fails. Indeed, we find this to be the case. We illustrate this for the Fixed Plateau in Fig.~\ref{fig:GSdistanceSC-QA}; similar results hold for the other examples we have considered.

\item {\emph{There is a sharp change in the ground state of the adiabatic quantum evolution at the degeneracy point.}} Tunneling should be accompanied by a sharp change in the properties of the ground state at the degeneracy point as the state state shifts from being localized in one well to the other. We quantify this change by calculating the expectation value of the Hamming weight operator, defined as $\mathrm{HW} = \frac{1}{2} \sum_{i=1}^n \left( \ident - \sigma^z_i \right)$. 
We expect a discontinuity in the spin-coherent ground state expectation value $\wich{\textrm{HW}}$, because the spin-coherent ground state changes discontinuously at the degeneracy point. We find that there is a nearly identical change in the quantum ground state expectation value $\wich{\textrm{HW}}$, for all of the examples listed above. This is illustrated explicitly for the Fixed Plateau in Fig.~\ref{fig:GSGibbs}. In Fig.~\ref{fig:HWdropsWellLocs}, we show that there is close and increasing agreement (as a function of $n$) between the position of the sudden drop in $\wich{\textrm{HW}}$ and the position of the degeneracy point, for all of the problems considered.

\item {\emph{The scaling of the barrier height in the spin-coherent potential is positively correlated with the scaling of the minimum gap of the quantum Hamiltonian.}} In Fig.~\ref{fig:heightsandgaps}, we see that as the barrier height increases, the inverse of the quantum minimum gap also increases. 

\end{enumerate}

Note that the Reichardt class is absent from the discussion above. The reason is that for these problems, the barrier in the spin-coherent potential is very small, which makes its numerical detection difficult. Fortunately, we can make some analytical claims about this class of problems. By adapting Reichardt's proof (reviewed in Appendix~\ref{app:review}) that these problems have a constant minimum gap, we are able to prove that the barrier height in the spin-coherent potential for this class vanishes as $n \to \infty$. Therefore, for these easy-for-AQA problems, there is a vanishing barrier in the spin-coherent potential. More precisely, we can show, for any perturbed Hamming weight problem, 
\begin{align}
V_{{\mathrm{SC}}}^{\mathrm{pert}} - V_{{\mathrm{SC}}}^{\mathrm{unpert}} &= s \sum_{l<k<u} f(k) \binom{n}{k} p(\theta)^k (1-p(\theta))^{n-k} \notag \\
&= \mathcal{O}\left(h \frac{u-l}{\sqrt{l}}\right)
\end{align}
where the unperturbed case refers to $h(n)=0$ in Eq.~\eqref{eqt:Reichardt}. 
Recall that $h \frac{u-l}{\sqrt{l}}$ = $o(1)$ for the Reichardt class. Thus asymptotically, the spin-coherent potential for this class approaches the spin-coherent potential of the unperturbed Hamming weight problem. It is easy to check that the latter has a single minimum throughout the evolution, and hence no barriers.

\begin{figure*}[htbp]
 \subfigure[]{\includegraphics[width=0.32\textwidth]{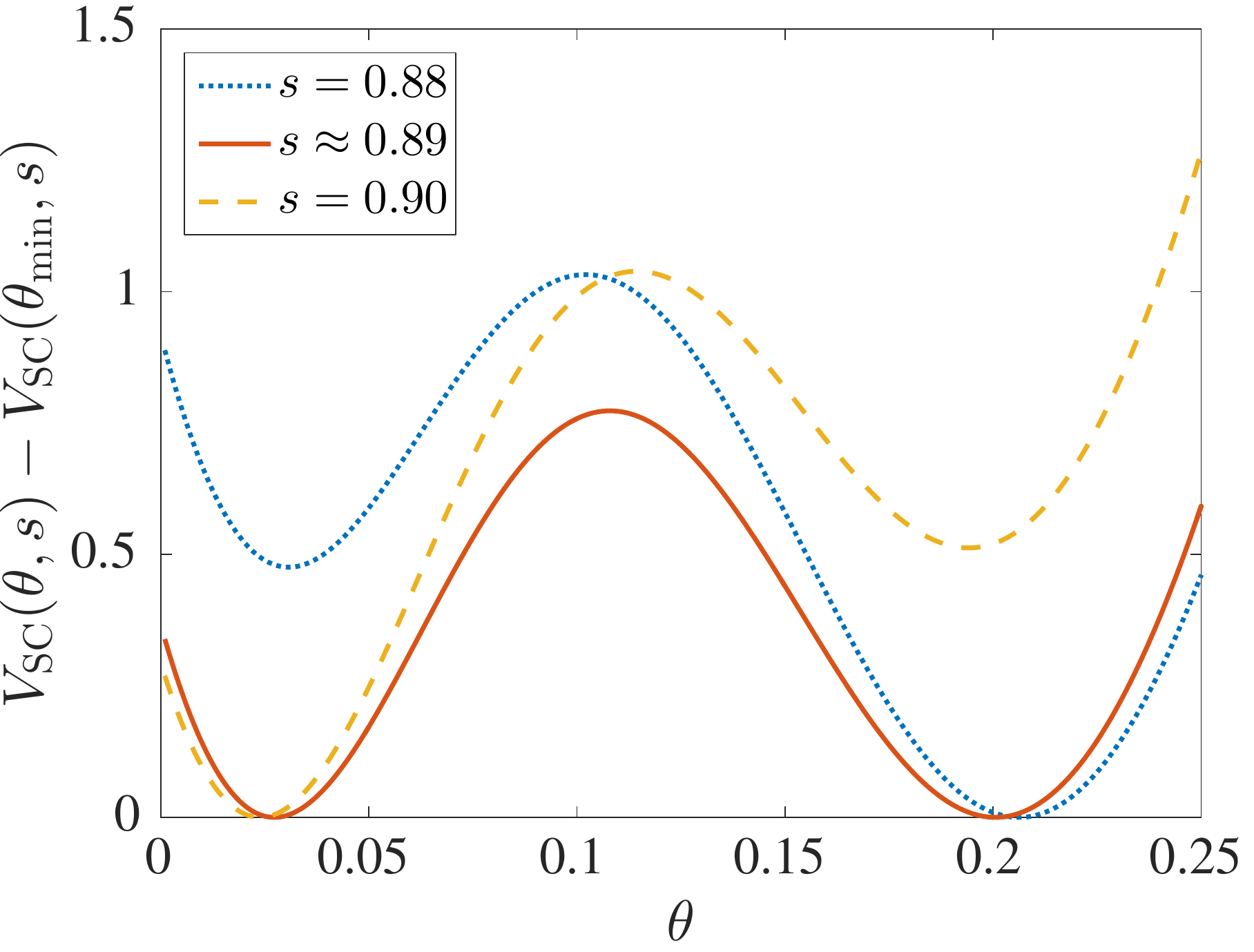} \label{fig:Veff}}
 \subfigure[]{\includegraphics[width=0.32\textwidth]{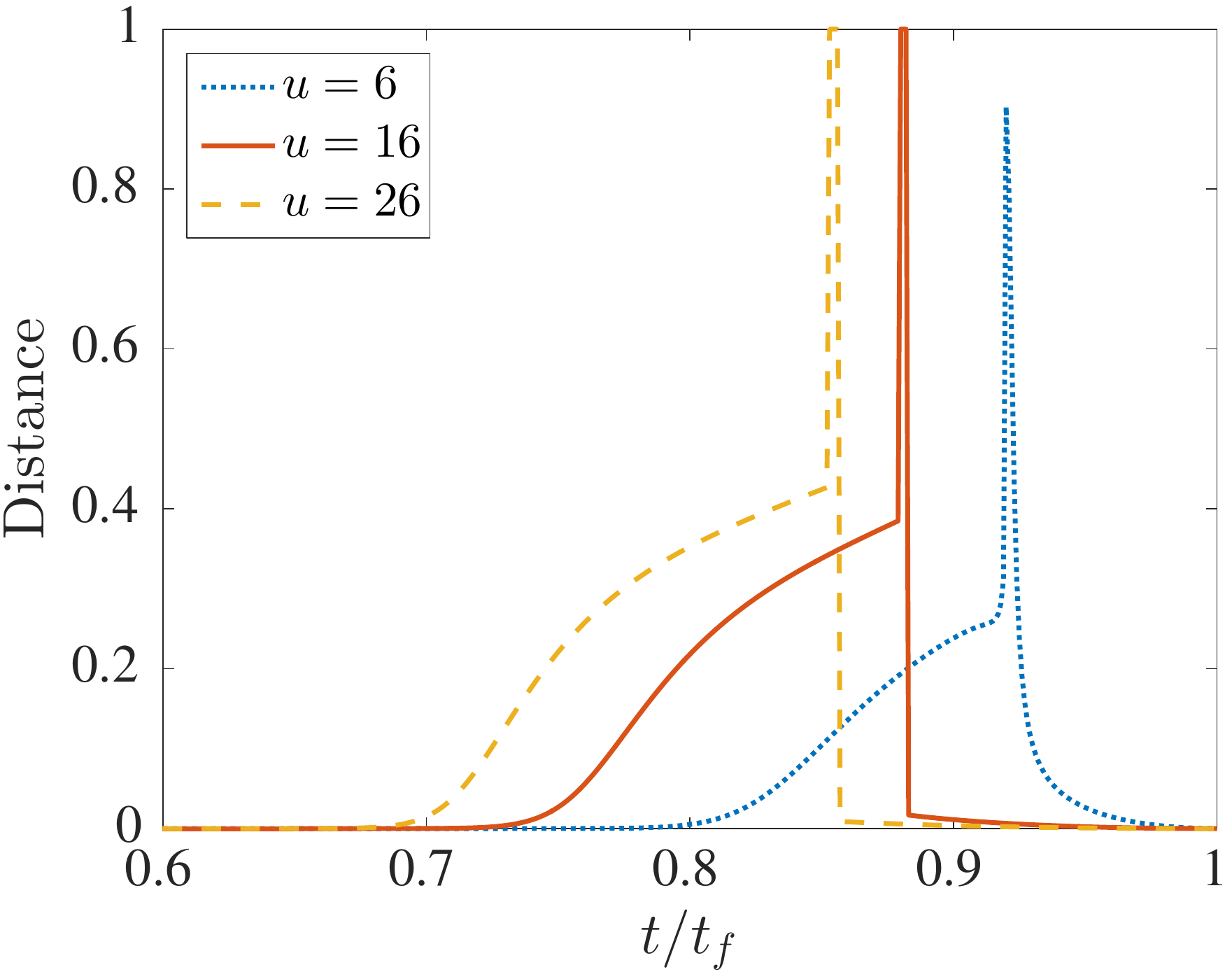} \label{fig:GSdistanceSC-QA}}
 \subfigure[]{\includegraphics[width=0.32\textwidth]{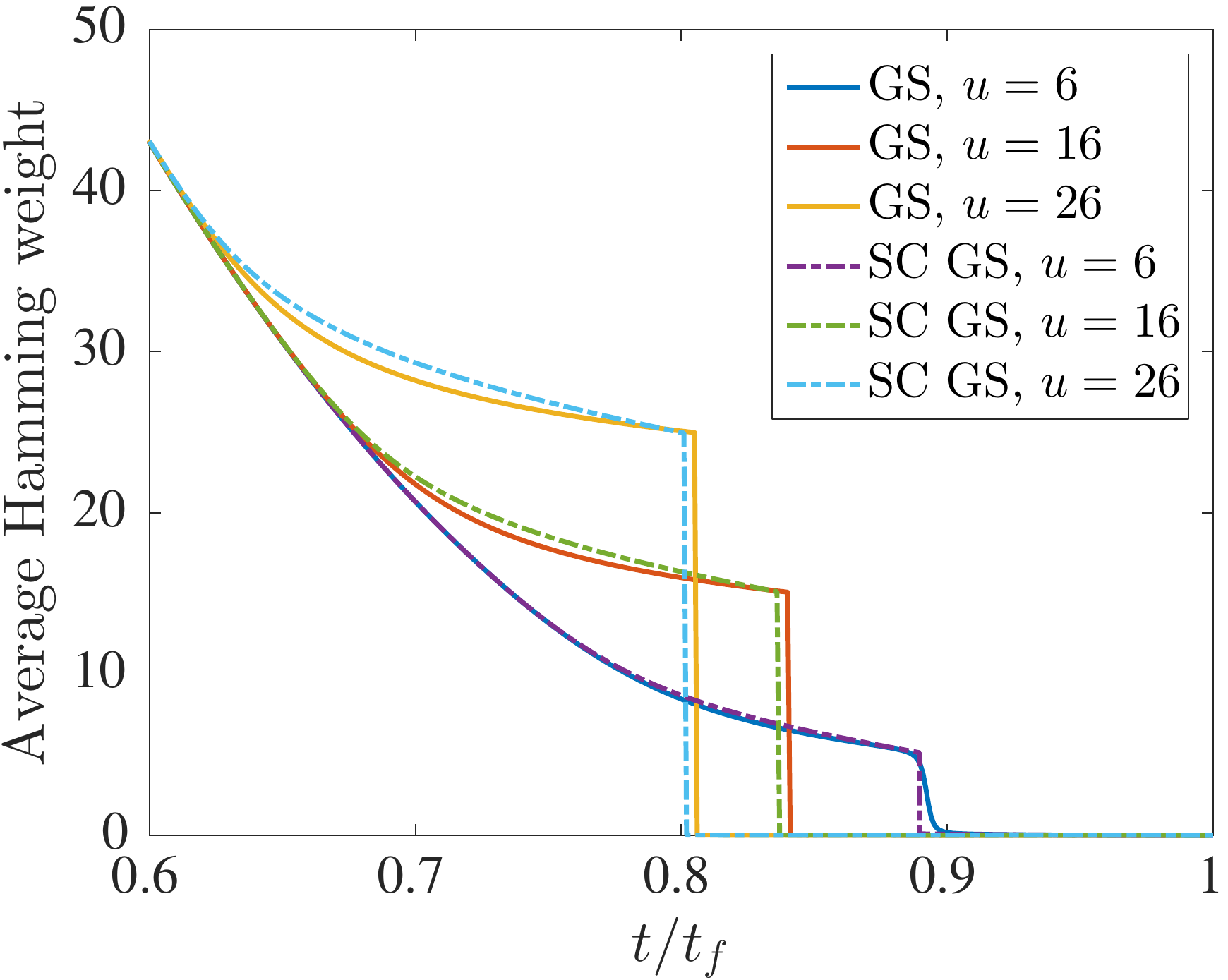} \label{fig:GSGibbs} }
 \caption{Results for the Fixed Plateau problem with $l=0$ and $n=512$. (a) The semi-classical potential with $u=6$ exhibits a double-well degeneracy at the position $s \approx 0.89$ (solid), but is non-degenerate before and after this point (dotted and dashed), thus leading to a discontinuity in the position of its global minimum. The same is observed for other the PHWO problems we studied (not shown). (b) The trace-norm distance between the quantum ground state (obtained by numerical diagonalization) and the spin-coherent state corresponding to the instantaneous global minimum in $V_{\mathrm{SC}}$, as a function of $t/t_f$. The peak corresponds to the location of the tunneling event, at which point the semi-classical approximation breaks down. (c) $\wich{\mathrm{HW}}$ in the instantaneous quantum ground state state (GS) and the instantaneous quantum ground state as predicted by the semi-classical potential (SC GS), as a function of $t/t_f$. The sharp drop in the GS and SC GS curves is due to a tunneling event wherein $\sim u$ qubits are flipped which occurs at the degeneracy point observed in the spin-coherent potential.}
 \end{figure*}

\begin{figure*}[t] 
	\subfigure[]{\includegraphics[width=0.4\textwidth]{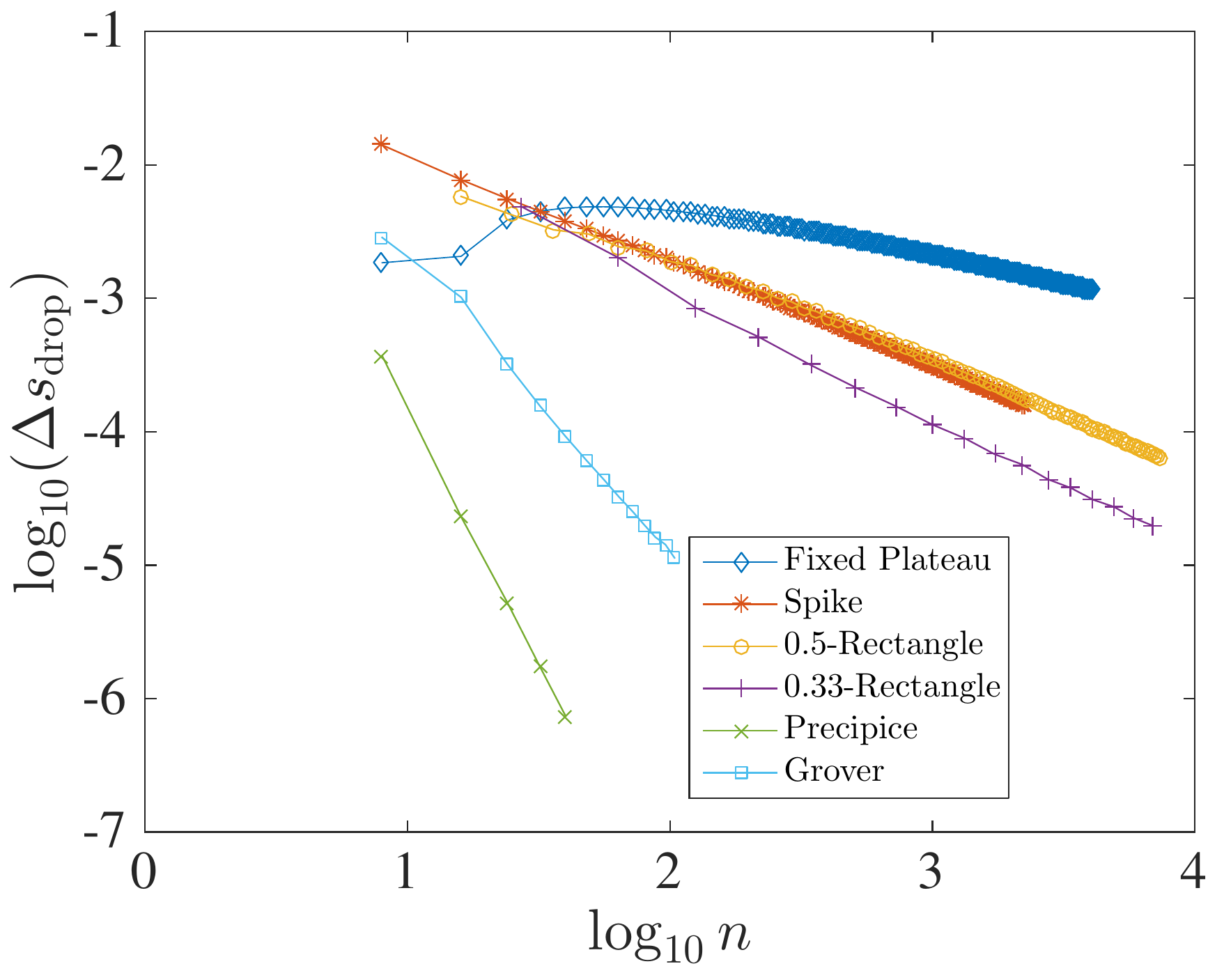} \label{fig:MinGapsWellLocs}} 
	\subfigure[]{\includegraphics[width=0.4\textwidth]{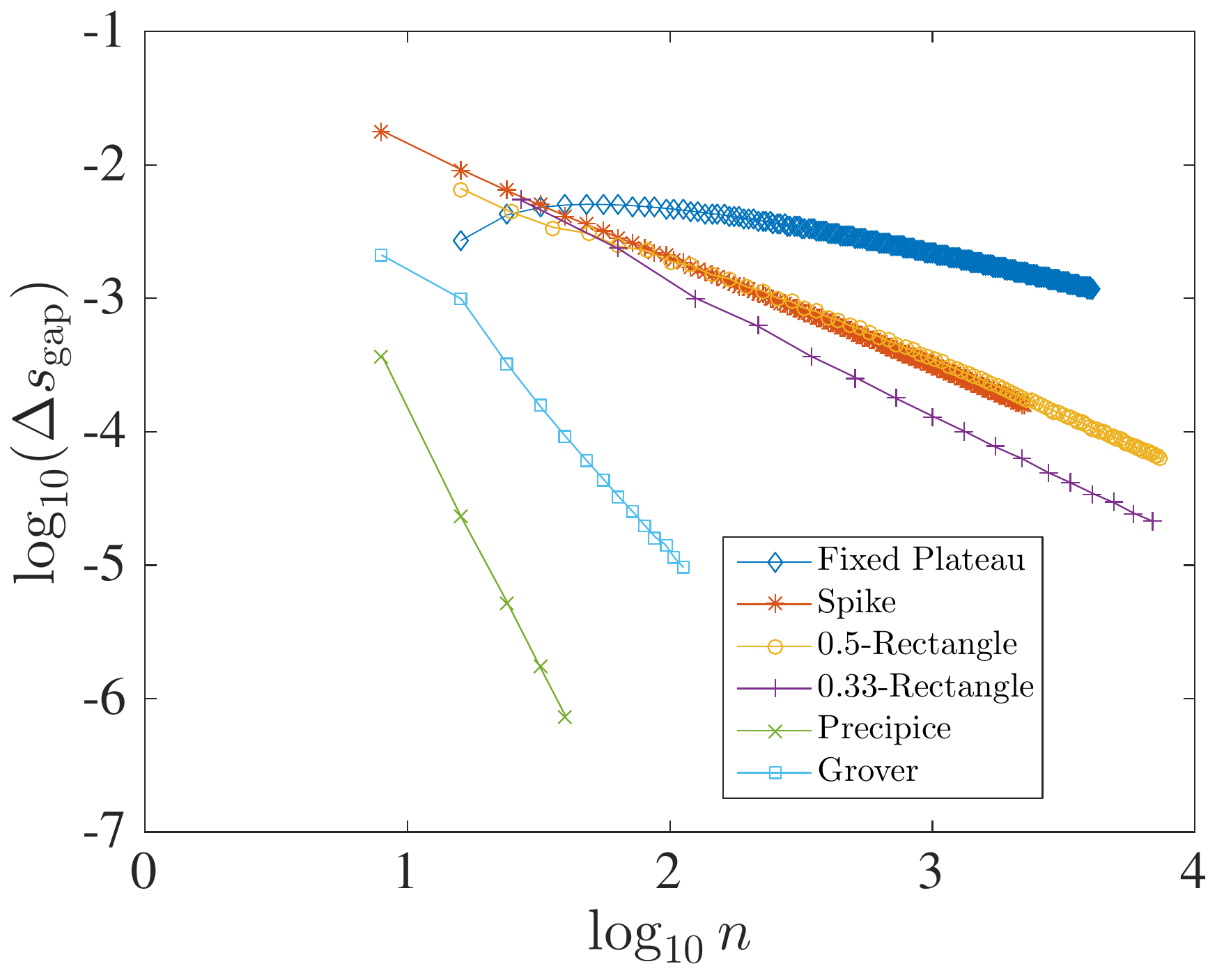} \label{fig:HWdropsWellLocs}} 
	\caption{(a) The difference in the position of the minimum gap from exact diagonalization and the position of the double-well degeneracy (as seen in Fig.~\ref{fig:Veff} for the Fixed Plateau) from the semi-classical potential, as a function of $n$ for the {Fixed Plateau}, the {Spike}, the {$0.5$-Rectangle}, the {$0.33$-Rectangle}, the {Precipice}, and {Grover} (log-log scale). (b) The difference in the position of the sudden drop in Hamming weight [as seen in Fig.~\ref{fig:GSGibbs}] and the position of the double-well degeneracy from the semi-classical potential [as seen in Fig.~\ref{fig:Veff}], as a function of the number of qubits $n$ for the {Fixed Plateau}, the {Spike}, the {$0.5$-Rectangle}, the {$0.33$-Rectangle}, the {Precipice}, and {Grover} (log-log scale). }
\end{figure*}

\begin{figure*}[t] 
	\subfigure[]{\includegraphics[width=0.4\textwidth]{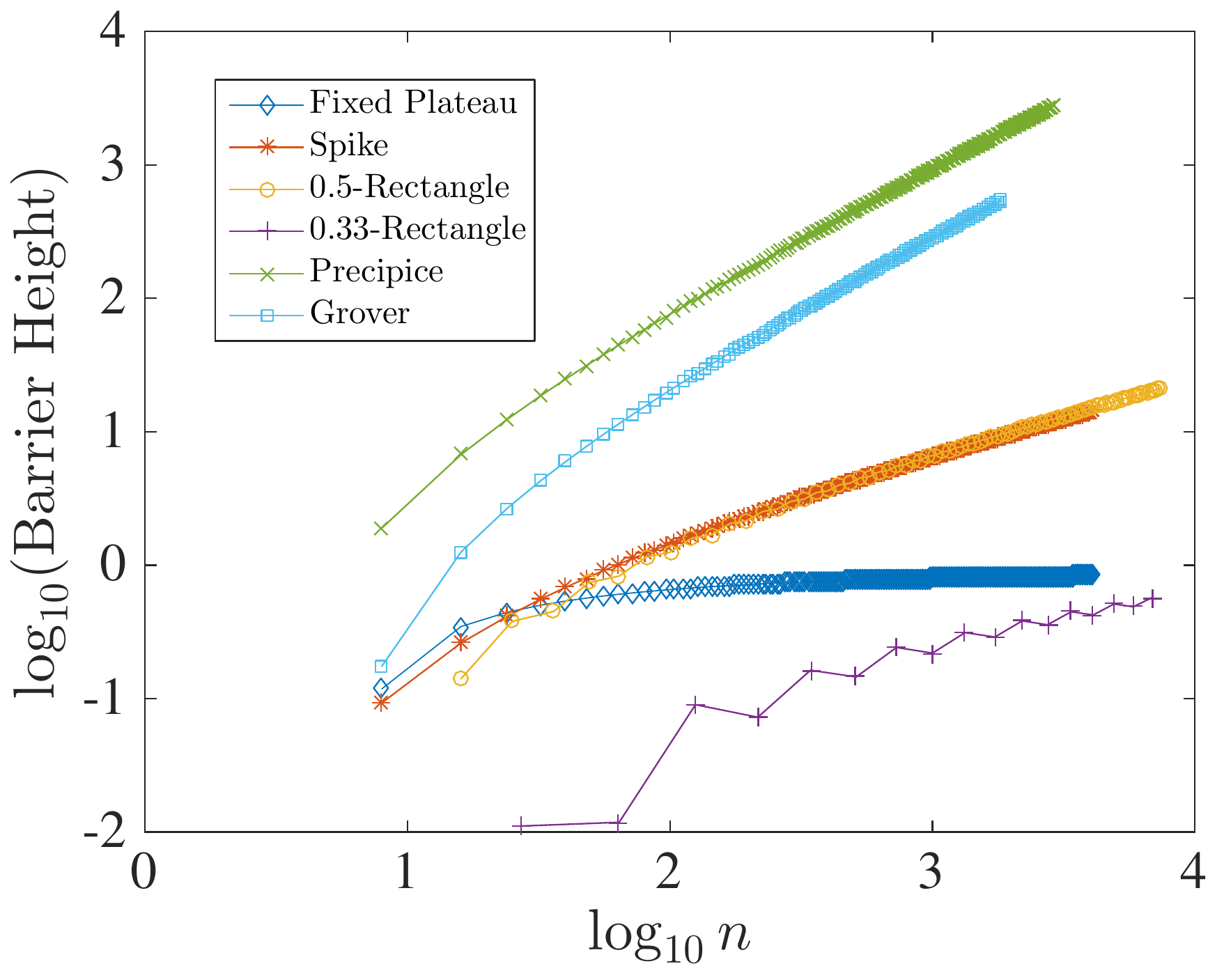} \label{fig:BarrierHeights}} 
	\subfigure[]{\includegraphics[width=0.4\textwidth]{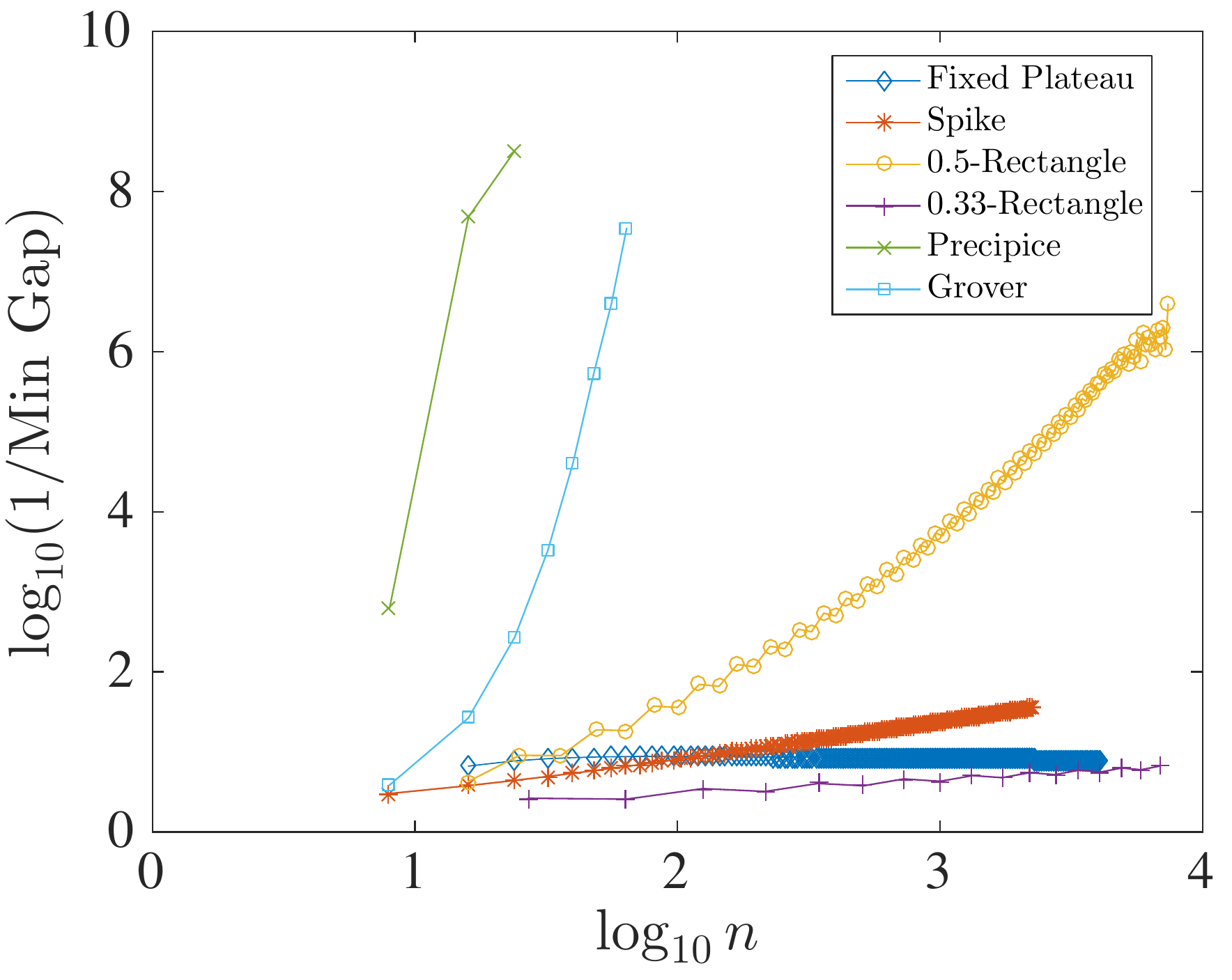} \label{fig:MinGaps}} 
	\caption{(a) The height of the barrier between the two wells at the degeneracy point of the spin-coherent potential [as seen in Fig.~\ref{fig:Veff}], as a function of $n$ for the {Fixed Plateau}, the {Spike}, the {$0.5$-Rectangle}, the {$0.33$-Rectangle}, the {Precipice}, and {Grover} (log-log scale). (b) The inverse of the minimum gap as a function of $n$ for the {Fixed Plateau}, the {Spike}, the {$0.5$-Rectangle}, the {$0.33$-Rectangle}, the {Precipice}, and {Grover} (log-log scale). The important thing to note is that the \emph{order} of scaling is preserved in both plots. That is, steeper the scaling of the barrier height, steeper the scaling of the inverse minimum gap.}
	\label{fig:heightsandgaps}
\end{figure*}

Taken together, these observations indicate that the spin-coherent potential (not the cost function alone) is the appropriate potential with respect to which tunneling is to be understood for these problems. 

\section{Fixed Plateau: Performance of algorithms} \label{sec:FixedPlateau}

Having motivated the spin-coherent potential for understanding tunneling, we now exhaustively analyze the {Fixed Plateau}. We choose this problem because it forces us to confront some intuitions about the performance of certain algorithms. Considering the final cost function, the Fixed Plateau has neither local minima nor a barrier going from large to small $|x|$: it just has a long, flat section before the ground state at $|x|=0$. This might suggest that it is easy for an algorithm such as SA, and is not a candidate for a quantum speedup. Moreover, given the absence of a barrier, one might suspect that the quantum evolution would not even involve multi-qubit tunneling. 

We dispel both of these intuitions and summarize our findings first.  In the previous section, we  already provided evidence that tunneling is unambiguously present for this problem.  The spin-coherent potential involves energy barriers, despite their absence in the final cost function, and the adiabatic quantum evolution is forced to tunnel in order to follow the ground state.  By a simulation of the Schr\"odinger equation, we find that AQA needs a time of $\mathcal{O}(n^{0.5})$ in order to reach a given success probability (see Sec.~\ref{subsec:adDyn}). Therefore, the adiabatic algorithm, via tunneling, is able to solve this problem efficiently.

Turning to SA, an algorithm which performs a local stochastic search on the final cost function, we prove that simulated annealing with single spin-updates will take time $\mathcal{O}(n^{u-l-1}) = \mathcal{O}(n^\mathrm{plateau\ width})$ to find the ground state (see Sec.~\ref{subsec:SArandom}). This result is due to the fact that a random walker on the plateau has no preferred direction and becomes trapped there. More precisely, the probability of a leftward transition while on the plateau is proportional to the probability of flipping one of a constant number of bits (given by the Hamming weight) out of $n$, which scales as $\sim 1/n$ if $l,u = \mathcal{O}(1)$. And since the walker needs to make as many consecutive leftward transitions as the width of the plateau in order to fall off the plateau, the time taken for this to happen scales as $\mathcal{O}(n^\mathrm{plateau\ width})$. Consequently, we obtain a polynomial speedup of AQA over SA that can be made as large as desired. Therefore, using the Fixed Plateau, we are able to demonstrate that a quantum speedup over SA is possible via tunneling in the adiabatic regime. 

However, is the adiabatic evolution optimal?  In order to find the optimal evolution time, we employ the optimal time to solution (TTS$_\mathrm{opt}$), a metric that is commonly used in benchmarking studies \cite{speedup} (also see Appendix~\ref{app:TTS}). It is defined as the minimum total time such that the ground state is observed at least once with desired probability $p_d$:
\beq 
\label{eqt:TTSopt}
\text{TTS}_\mathrm{opt} = \min_{t_f > 0} \left\{ t_f \frac{\ln( 1- p_d)}{\ln \left(1-p_{\mathrm{GS}}(t_f) \right)} \right\} \ ,
\eeq
where $t_f$ is the duration (in QA) or the number of single spin updates (in SA) of a single run of the algorithm, and $p_{\mathrm{GS}}(t_f)$ is the probability of finding the ground state in a single such run. The use of TTS$_\mathrm{opt}$ allows for the possibility that multiple short runs of the evolution, each lasting an optimal annealing time $({t_f})_{\textrm{opt}}$, result in a better scaling than a single long (adiabatic) run with an unoptimized $t_f$.  The quantum evolution that gives the optimal annealing time relative to this cost function is actually DQA, with an asymptotic scaling of $\mathcal{O}(1)$. Importantly, this diabatic evolution does not contain any of the signatures of tunneling  discussed in the previous section.  Therefore, for the Fixed Plateau, tunneling does not give rise to the optimal quantum performance. 

Motivated by the fact that the optimal quantum evolution involves no multi-qubit tunneling, we consider spin-vector dynamics \cite{Smolin} (see, also Refs.~\cite{Albash:2014if,owerre2015macroscopic}), a model that evolves according to the spin-coherent potential in Eq.~\eqref{eq:vsc}.  SVD can be derived as the saddle-point approximation to the path integral formulation of QA in the spin-coherent basis \cite{owerre2015macroscopic}. The SVD equations are equivalent to the Ehrenfest equations for the magnetization under the assumption that the density matrix is a product state, i.e., $\rho = \otimes_{i=1}^n \rho_{i}$, where $\rho_i$ denotes the state of the $i$th qubit.  This algorithm is useful since it is derived under the assumption of continuity of the angles $(\theta, \phi)$, so tunneling, which here would amount to a discrete jump in the angles, is absent.  

We also consider a quantum Monte Carlo based algorithm, often called simulated quantum annealing (SQA) \cite{sqa1,Santoro}. We show that SQA has a scaling that is better than SA's. Indeed, this is consistent with the fact that SQA thermalizes not just relative to the final cost function, but also during the evolution.

We provide further details of our implementations of each of these algorithms in Appendix \ref{app:Methods}. We now turn to each of the algorithms individually and detail their performance for the Fixed Plateau problem.

\subsection{Adiabatic dynamics} \label{subsec:adDyn}
In order to study the scaling of adiabatic dynamics, we consider the minimum time $\tau_{0}$ required to reach the ground state with some probability $p_0$, where we choose $p_0$ to ensure that we are exploring a regime close to adiabaticity for QA.   We call this benchmark metric the ``threshold criterion,'' and set  $p_0=0.9$.  As seen in Fig.~\ref{fig:ThC}, we observe a scaling for AQA that is approximately $\sim n^{0.5}$. As is to be expected given that the tunneling for the Fixed Plateau problem is controlled by the width of the plateau, which is constant (does not scale with $n$),  we find that $\tau_{0}$ scales in the same way for the Fixed Plateau and the Plain Hamming Weight problems (see Appendix~\ref{app:review}). This  suggests that the dominant contribution to the scaling at large $n$ is not the time associated with tunneling but rather the time associated with the Plain Hamming Weight problem.

As also seen in Fig.~\ref{fig:ThC}, we find that the textbook adiabatic criterion \cite{Messiah:book} given by
\beq
t_f \gtrsim \max_{s \in [0,1]}  \frac{ |\bra{\eps_0(s) } \partial_s H(s) \ket{\eps_1(s)}|}{\text{Gap}(s)^2} \ ,
\label{eqt:AdC}
\eeq
serves as an excellent proxy for the scaling of AQA %
\footnote{We note that while this adiabatic criterion matches the numerical scaling we observe for the quantum evolution, it is well known to be neither exact nor general; see, e.g., Refs.~\cite{Jansen:07,Amin:09,lidar:102106}.}.
The scaling of AQA is matched by the scaling of the numerator of the adiabatic condition, which is explained by the fact that we find a constant minimum gap for the case $l,u = \mathcal{O}(1)$. This numerator turns out to be well approximated in our case by the matrix element of $H(s)$ between the ground and first excited states, leading to $t_f\sim n^{0.5}$ in the adiabatic limit. Note that calculating this matrix element can easily be done for arbitrarily large systems, and is hence much easier to check directly than the scaling of AQA.

\begin{figure}[t] 
\includegraphics[width=0.8 \columnwidth]{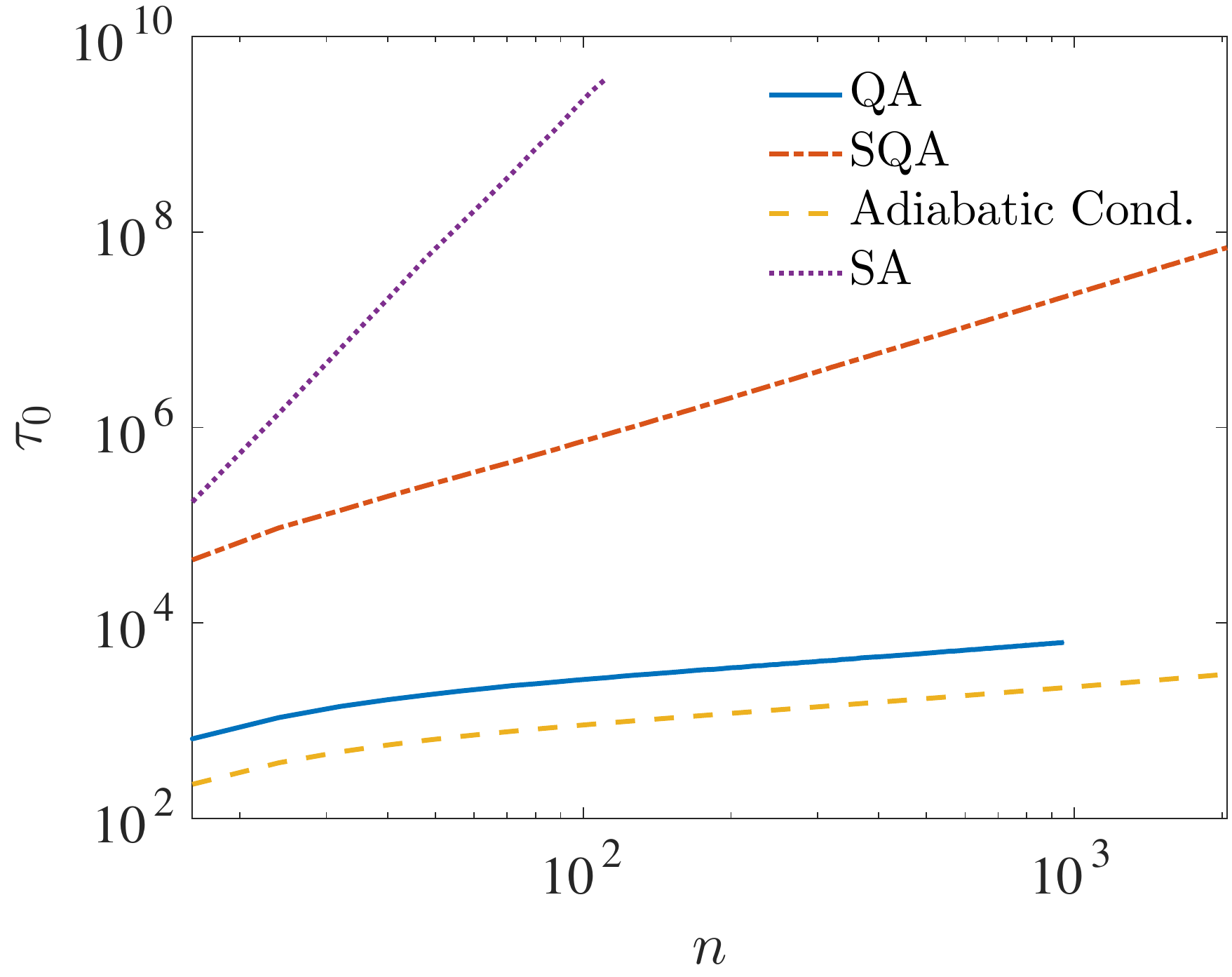}
\caption{Performance of different algorithms for the Fixed Plateau problem with $l=0$ and $u =6$. Shown is a log-log plot of the scaling of the time to reach a threshold success probability of $0.9$, as a function of system size $n$ for AQA, SQA ($\beta = 30$, $N_{\tau} = 64$) and SA ($\beta_f = 20$).  The time for SQA and SA is measured in single-spin updates (for SQA this is $N_{\tau}$ times the number of sweeps times the number of spins, whereas for SA this is the number of sweeps times the number of spins), where both are operated in `solver' mode as described in Appendix~\ref{app:Methods}.  Also shown is the scaling of the numerator of the adiabatic condition as defined in Eq.~\eqref{eqt:AdC}. The scaling for AQA and the adiabatic condition extracted by a fit using $n \gtrsim 10^2$ is approximately $n^{0.44}$. However, the true asymptotic scaling is likely to be $\sim n^{0.5}$ since the scaling for the Fixed Plateau problem is clearly lower-bounded by the Plain Hamming Weight problem, for which we have verified $\tau_0 \sim n^{0.5}$ (see Appendix~\ref{app:review}), and we expect the effect of the plateau to become negligible in the large $n$ limit.  SQA scales more favorably ($\sim n^{1.5}$) than SA ($\sim n^{5}$).  We have checked that the scaling of SQA does not change even if we double the number of Trotter slices $N_{\tau}$ and keep the temperature $1/\beta$ fixed.} 
\label{fig:ThC} 
\end{figure}
%
\subsection{Simulated annealing using random spin selection} \label{subsec:SArandom}
%
We consider a version of SA with random spin-selection as the rule that generates candidates for Metropolis updates. Our main motivation is to understand the behavior of a local, stochastic search algorithm which has access only to the final cost function. We note that our analysis below is general for any Plateau problem, and is not limited to the {Fixed Plateau} or the {Moving Plateau}.

If we pick a bit-string at random, then for large $n$ we will start with very high probability at a bit-string with Hamming weight close to $n/2$.  The plateau may be to the left or to the right of $n/2$; if the plateau is to the right, then the random walker is unlikely to encounter it and can quickly descend to the ground state.
Thus, the more interesting case is when the random walker arrives at the plateau from the right. We proceed to analyze these two cases separately.

\subsubsection{Walker starts to the right of the plateau}
In this case, how much time would it take, typically, for the walker to fall off the left edge? It is intuitively clear that traversing the plateau will be the dominant contribution to the time taken to reach the ground state, as after that the random walker can easily walk down the potential. We show below (for the walker that starts to the left of the plateau) that this time can be at most $\mathcal{O}(n^2)$ if $\beta = \Omega(\log n)$.

To evaluate the time to fall off the plateau, note that the perturbation is applied on strings of Hamming weight $l+1,l+2,\dots,u-1$, so the width of the plateau is $w \equiv u-l-1$. Consider a random walk on a line of $w+1$ nodes labelled $0,1,\dots w$. Node $i$ represents the set of bit strings with Hamming weight $l+i$, with $0\leq i \leq w$. We may assume that the random walker starts at node $w$.
Only nearest-neighbor moves are allowed and the walk terminates if the walker reaches node $0$.

Our analysis will provide a lower bound on the actual time to fall off the left edge, because in the actual PHWO problem one can also go back up the slope on the right, and in addition we disallow transitions from strings of Hamming weight $l$ to $l+1$. This is justified because the Metropolis rule exponentially (in $\beta$) suppresses these transitions.

The transition probabilities $p_{i \to j}$ for this problem can be written as a $(w+1) \times (w+1)$ row-stochastic matrix $p_{ij} = p_{i \to j}$. Here $p$ is a tridiagonal matrix with zeroes on the diagonal, except at $p_{00}$ and $p_{ww}$.  First consider $1\leq i \leq w-1$. If the walker is at node $i$, then the transition to node $i+1$ (which has Hamming weight $l+i+1$) occurs with probability $\frac{n-(l+i)}{n}$ (the chance that the bit picked had the value $0$). Similarly, for $1\leq i \leq w$, the Hamming weight will decrease to $l+i-1$ with probability $\frac{l+i}{n}$ (the chance that the bit picked had the value $1$). Combining this with the fact that a walker at node $0$ stays put, we can write:
\bes
\begin{align}
b_i &\equiv p_{i\to i}=\begin{cases} 1 \text{  if  } i=0 \\
								  0 \text{  if  } 1\leq i \leq (w-1) \\ 
								  1-\frac{l+w}{n} \text{  if  } i=w 
								  \end{cases},\\		  
c_i &\equiv p_{i-1\to i} = \begin{cases} 0 \text{  if  } i=1\\ 1-\frac{l+i-1}{n} \text{  if  } i=2,\dots,w \end{cases}, \\
a_i &\equiv p_{i\to i-1} = \frac{l+i}{n} \text{  if  } i=1,2,\dots,w.
\end{align}
\ees
Let $X(t)$ be the position of the random walker at time-step $t$. The random variable measuring the number of steps the random walker starting from node $r$ would need to take to reach node $s$ for the first time is
\beq
\tau_{r,s} \equiv \text{min} \{ t>0: X(t)=s,X(t-1)\neq s | X(0) =r\}\ .
\eeq
The quantity we are after is $\mathbb{E}\tau_{w,0}$, the expectation value of the random variable $\tau_{w,0}$, i.e., the mean time taken by the random walker to fall off the plateau. Since only nearest neighbor moves are allowed we have
\beq \label{eq:timetofalloff}
\mathbb{E}\tau_{w,0} =\sum_{r=1}^w \mathbb{E}\tau_{r,r-1}\ .
\eeq
Stefanov~\cite{stefanov1995mean} (see also Ref.~\cite{krafft1993mean}) has shown that
\beq \label{eq:stefanovformula}
\mathbb{E}\tau_{r,r-1} = \frac{1}{a_r} \left( 1 + \sum_{s=r+1}^{w} \prod_{t=r+1}^{s} \frac{c_t}{a_t} \right),
\eeq
where $c_{w+1}\equiv 0$. Evaluating the sum term by term, we obtain: 
\bes
\label{eq:62}
\begin{align}
\mathbb{E}\tau_{w,w-1} &= \frac{n}{l+w}, \\
\vdots \nonumber \\
\mathbb{E}\tau_{w-k,w-k-1} &= \frac{n}{l+w-k} \left[1+ \frac{n-(l+w-k)}{l+w-(k-1)}+\dots \right. \nonumber \\
&  +\frac{n-(l+w-k)}{l+w-(k-1)}\times\cdots  \nonumber \\ 
& \left. \times\frac{n-(l+w-2)}{l+w-1} \times \frac{n-(l+w-1)}{l+w}\right]. 
\end{align}
\ees
Now consider the following cases:
\begin{enumerate}
\item \emph{Fixed Plateau}, $l,u = \mathcal{O}(1)$: Here, using the fact that $k=\mathcal{O}(w)=\mathcal{O}(1)$, we conclude that $\mathbb{E}\tau_{w-k,w-k-1} = \mathcal{O}(n^{k+1})$. Since the leading order term is $\mathbb{E}\tau_{w-(w-1),w-w}=\mathbb{E}\tau_{1,0}$, the time to fall off the plateau is $\mathcal{O}(n^w) = \mathcal{O}(n^{u-l-1}).$ This result about SA's performance is confirmed numerically in Fig.~\ref{fig:ThC}.

\item In order for Reichardt's bound (see Appendix~\ref{app:review}) to give a constant lower-bound to the quantum problem, we need $u=l+o(l^{1/4})$. Since at most we can have $l = \mathcal{O}(n)$, we can conclude $\mathbb{E}\tau_{w-k,w-k-1} = \mathcal{O}\left(\frac{n}{l}\right)^{k+1}$. Therefore, the time to fall-off becomes $\mathbb{E}\tau_{w,0} = \mathcal{O}\left(w(\frac{n}{l})^w\right)$. 
\begin{itemize}
\item \emph{Moving Plateau}: If $l = \Theta(n)$ and $w = \mathcal{O}(1)$, we can see that $\mathbb{E}\tau_{w,0} = \mathcal{O}(1)$, which is a constant time scaling.
\item \emph{Moving Plateau with changing width}: If $l = \Theta(n)$ and $w = \mathcal{O}(n^a)$, where $0<a<1/4$, then $\mathbb{E}\tau_{w,0} = \mathcal{O}(n^a \mathcal{O}(1)^{n^a})$, which is super-polynomial.
\item \emph{Most general plateau in the Reichardt class}: More generally, if $l = \mathcal{O}(n^b)$, with $b \leq 1$ and $w = \mathcal{O}(n^a)$, where $0\leq a < b/4$, then we get the scaling $\mathbb{E}\tau_{w,0} = \mathcal{O}(n^{a}\mathcal{O}(n^{1-b})^{n^a})$
\end{itemize}
\end{enumerate}

\subsubsection{Walker starts to the left of the plateau}
\label{subsubsec:SAHamWt}
Note that this case is equivalent to the unperturbed Hamming weight problem, which is a straightforward gradient descent problem. We may therefore consider a simple fixed temperature version of SA (i.e., the standard Metropolis algorithm). We will show that the performance of SA on this problem provides an upper bound of $\mathcal{O}(n^2)$ on the time for a random walker to arrive at the plateau, and on the time for a random-walker to reach the ground state after descending from the plateau.  Moreover, our analysis provides a lower bound of $\mathcal{O}(n\log n)$ on the efficiency of such algorithms.

For this problem, the transition probabilities are:
\bes
\begin{align}
c_i &\equiv p_{i-1\to i} = \frac{n-i+1}{n} e^{-\beta} \ , \\
a_i &\equiv p_{i\to i-1} =  \frac{i}{n} \ ,
\end{align}
\ees
with $i=1,2,\dots,n$ denoting strings of Hamming weight $i$, and $\beta$ is the inverse temperature. Using the Stefanov formula~\eqref{eq:stefanovformula}, we can write (after much simplification):
\beq
\mathbb{E}\tau_{n-k,n-k-1} = \frac{n}{n-k} \binom{n}{k}^{-1} \sum_{l=0}^{k} e^{-l \beta} \binom{n}{k-l} \ .
\eeq
We will bound 
\beq \label{eqt:hamwttime}
\mathbb{E}\tau_{n,0} = \sum_{k=0}^{n-1} \frac{n}{n-k} \binom{n}{k}^{-1} \sum_{l=0}^{k} e^{-l \beta} \binom{n}{k-l} \ ,
\eeq
the expected time to reach the all-zeros string starting from the all-ones string. This is the worst-case scenario as we are assuming that we are starting from the string farthest from the all-zeros string. Note again that if we start from a random spin configuration, then with overwhelming probability we will pick a string with Hamming weight close to $n/2$. Thus, most probably, $\mathbb{E}\tau_{n/2,0}$ will be the time to hit the ground state.

We first show that $\beta = \mathcal{O}(1)$ will lead to an exponential time to hit the ground state, irrespective of the walker's starting string. 
Toward that end,
\bes
\begin{align}
\mathbb{E}\tau_{1,0} &= \mathbb{E}\tau_{n-(n-1),n-n} \\
&=\sum_{l=0}^{n-1} e^{-l \beta} \binom{n}{n-1-l} \\
&=e^\beta\left[(e^{-\beta}+1)^n-1\right],
\end{align}
\ees
which is clearly exponential in $n$ if $\beta= \mathcal{O}(1)$. 

Next, let $\beta(n) = \log n$, i.e., we decrease the temperature logarithmically in system size. In this case,
\begin{align}
\mathbb{E}\tau_{1,0} = n \left[\left(1 + \frac{1}{n}\right)^n-1\right] \leq n (e - 1) = \mathcal{O}(n)\ .
\end{align}
Now it is intuitively clear that $\mathbb{E}\tau_{1,0} > \mathbb{E}\tau_{r,r-1}$ for all $r>1$, which implies that $n \mathbb{E}\tau_{1,0} \geq \mathbb{E}\tau_{n,0}  $. Thus, if $\beta = \log n$, then $\mathbb{E}\tau_{n,0} = \mathcal{O}(n^2)$ at worst.

To obtain a lower-bound on the performance of the algorithm, we take $\beta \to \infty$. Thus,  for each $k$ in Eq.~\eqref{eqt:hamwttime}, only the $l=0$ term will survive. Hence,
\begin{align}
\lim_{\beta \to \infty} \mathbb{E}\tau_{n,0} &= \sum_{k=0}^{n-1} \frac{n}{n-k} = n \sum_{i=1}^n \frac{1}{i} \approx n (\log n + \gamma)\ ,
\end{align}
for large $n$, with $\gamma$ being the Euler-Mascheroni constant. The scaling here is $\mathcal{O}(n \log n)$. This is the best possible performance for single-spin update SA with random spin-selection on the plain Hamming weight problem. Therefore, if $\beta = \Omega(\log n)$, the scaling will be between $\mathcal{O}(n \log n)$ and $\mathcal{O}(n^2)$. Of course, this cost needs to be added to the time taken for the walker starting to the right of the plateau.

Two clarifications are in order regarding the comparison between our theoretical bound on SA's performance and the associated numerical simulations we have presented. First, while Fig.~\ref{fig:ThC} displays the time to cross a threshold probability, our theoretical bound of $\mathcal{O}(n^{u-l-1})$ is on the expected time for the random walker to hit the ground state [Eq.~\eqref{eq:timetofalloff}]. However, we found that both metrics show identical scaling. Second, while the SA data in Fig.~\ref{fig:ThC} was generated using sequential spin updates, the theoretical bound assumes random spin updates (see Appendix~\ref{app:SA} for more details on the update schemes). However, we found that the asymptotic scaling for both cases is nearly identical in the long-time regime, and thus have plotted only the former.

\subsection{Optimal QA via Diabatic Transitions} \label{subsec:diabaticQA}

%
\begin{figure*}[t] 
   \centering
 \subfigure[]{\includegraphics[width=0.32\textwidth]{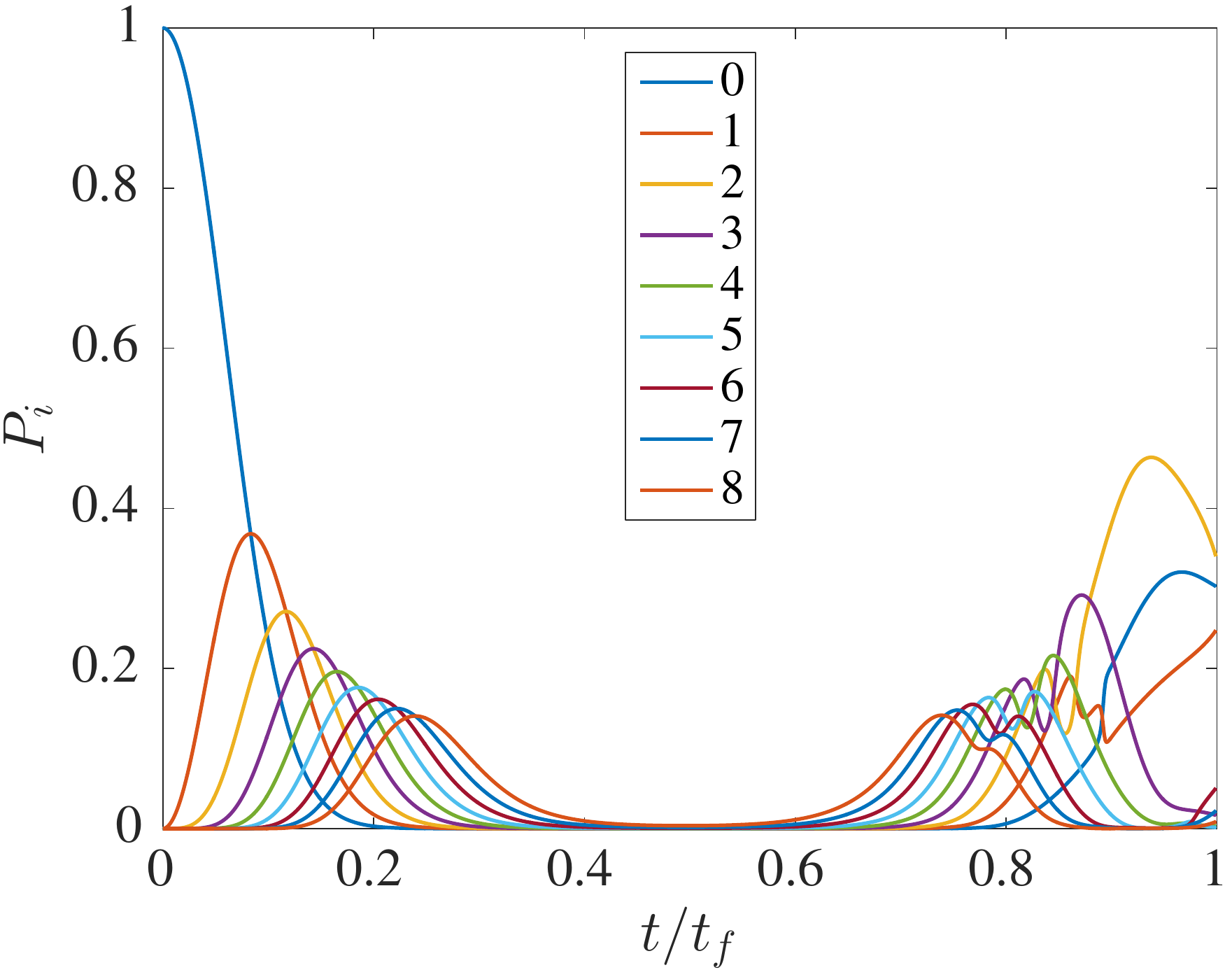}  \label{fig:QA_EnergyOverlap}}
   \subfigure[]{\includegraphics[width=0.32\textwidth]{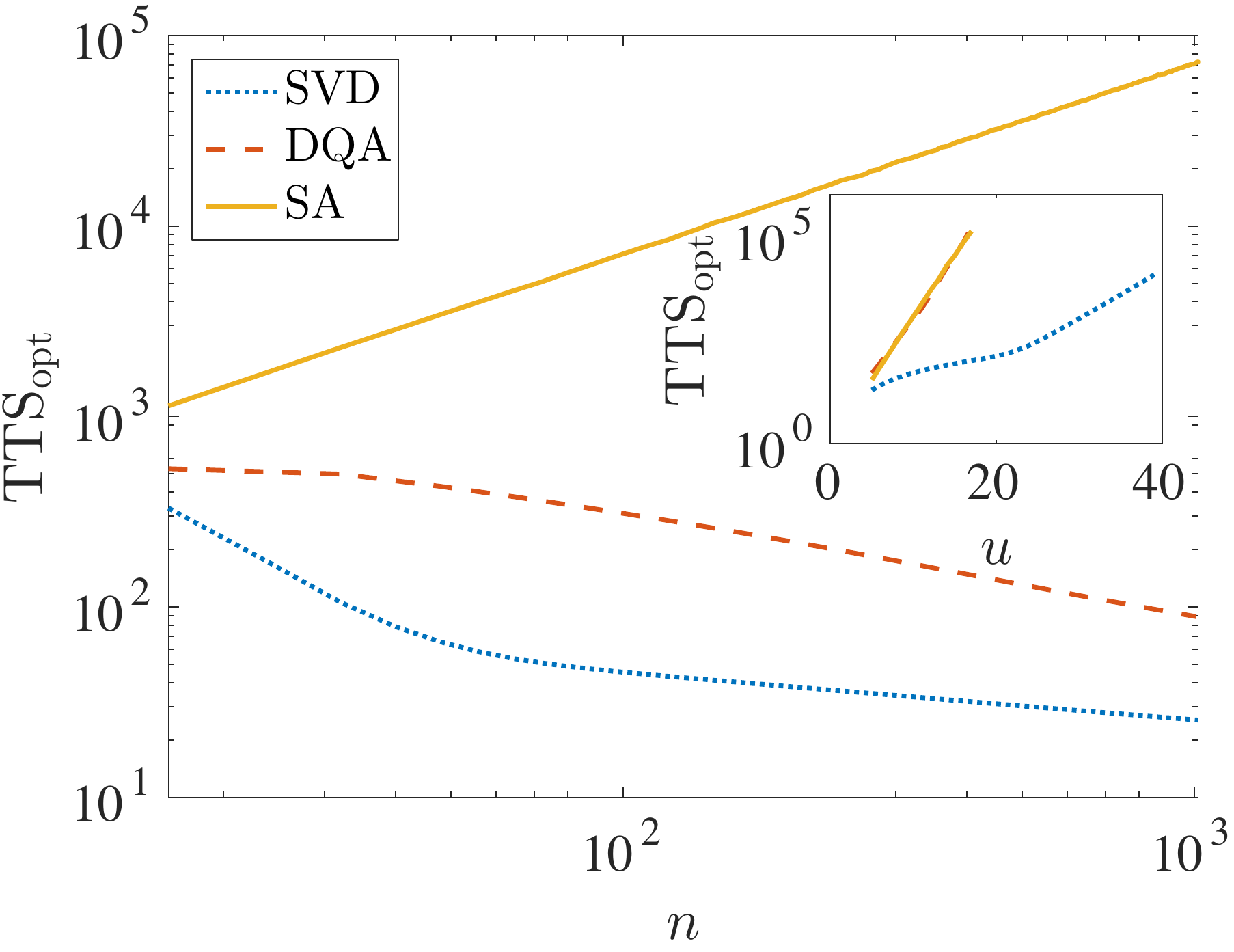} \label{fig:TTSScaling}}
  \subfigure[]{\includegraphics[width=0.32\textwidth]{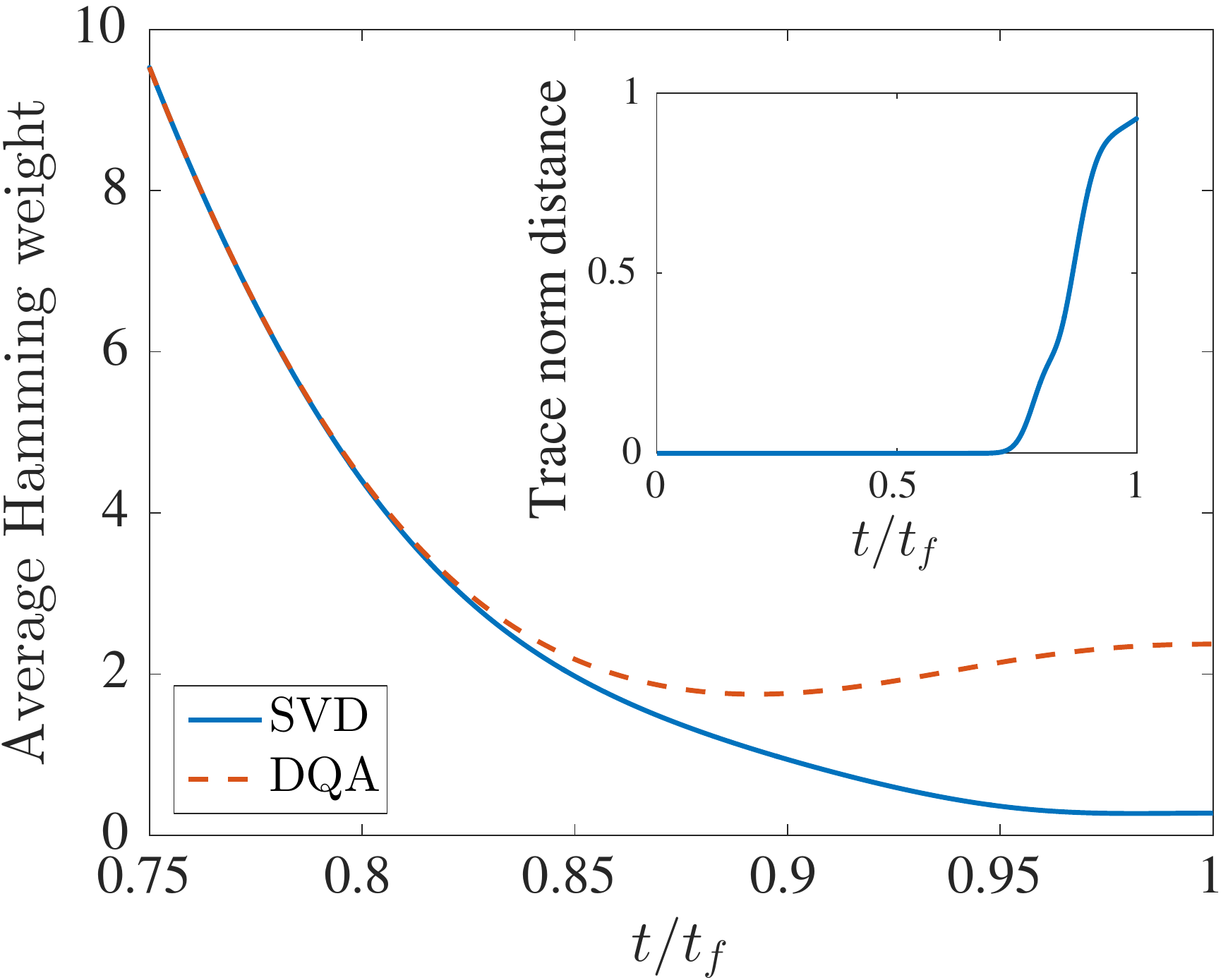}   \label{fig:QA_O3_AveHW}}
   \caption{Diabatic QA \textit{vs} SA and SVD for the {Fixed Plateau} problem with $l=0$. (a) Population $P_i$ in the $i$th energy eigenstate along the diabatic QA evolution at the optimal TTS for $n=512$ and $u=6$.  Excited states are quickly populated at the expense of the ground state.  By $t/t_f = 0.5$ the entire population is outside the lowest $9$ eigenstates.  In the second half of the evolution the energy eigenstates are repopulated in order. This kind of dynamics occurs due to a lining-up of avoided level crossings as seen in Fig.~\ref{fig:Cascade}. (b) Scaling of the optimal TTS with $n$ for $u=6$, with an optimized number of single-spin updates for SA, and equal $({t_f})_\textrm{opt}$ for DQA and SVD. SA scales as $\mathcal{O}(n)$, a consequence of performing sequential single-spin updates. DQA and SVD both approach $\mathcal{O}(1)$ scaling as $n$ increases. Here we set $p_d=0.7$ in Eq.~\eqref{eqt:TTSopt}, in order to be able to observe the saturation of SVD's TTS to the point where a single run suffices, i.e., TTS$_\textrm{opt} = ({t_f})_\textrm{opt}$. The conclusion is unchanged if we increase $p_d$: this moves the saturation point to larger $n$ for both SVD and DQA, and we have checked that SVD always saturates before DQA.
Inset: scaling as a function of $u$ for $n=1008$. SVD is again seen to exhibit the best scaling, while for this value of $n$ the scaling of DQA and SA is similar (DQA's scaling with $n$ improves faster than SA's as a function of $n$, at constant $u$). (c) $\wich{\mathrm{HW}}$ of the QA wavefunction and the SVD state (defined as the product of identical spin-coherent states) for $n=512$ and $u = 6$.  The behavior of the two is identical up to $t/t_f \approx 0.8$, when they begin to differ significantly, but neither displays any of the sharp changes observed in Fig.~\ref{fig:GSGibbs} for the instantaneous ground state.  Inset: the trace-norm distance between the DQA and SVD states, showing that they remain almost indistinguishable until $t/t_f \approx 0.8$.}
\end{figure*}
%
\begin{figure}[t] 
   \centering
   \includegraphics[width=0.8 \columnwidth]{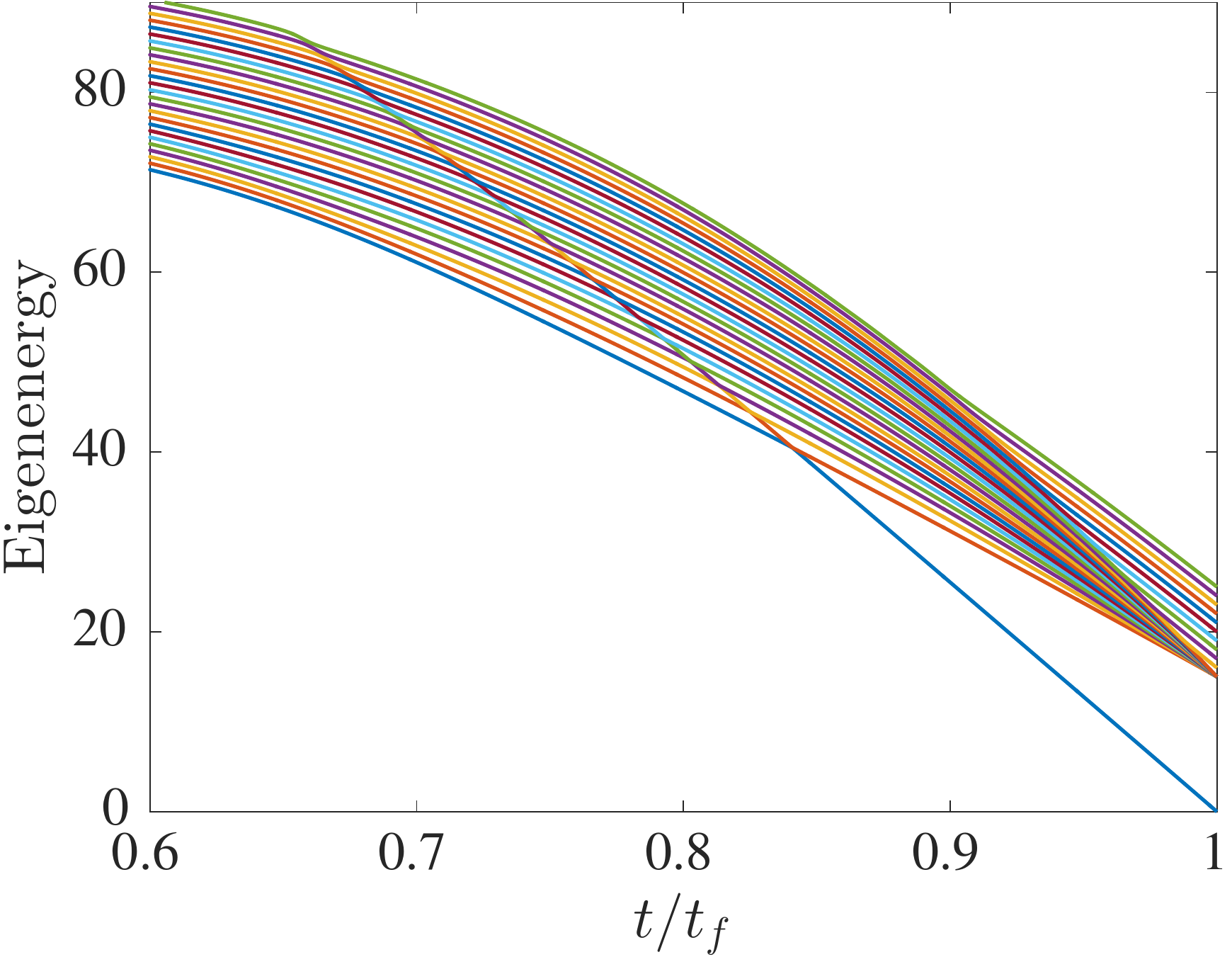}
      \caption{The eigenenergy spectrum along the evolution for the {Fixed Plateau} with $n=512$, $l=0$, and $u=6$. Note the sequence of  avoided level crossings that unmistakably line up in the spectrum to reach the ground state. This is the pathway through which DQA is able to achieve a speedup over AQA.}
       \label{fig:Cascade}
\end{figure}

Having established that for the Fixed Plateau AQA enjoys a quantum speedup over local search algorithms such as SA via tunneling, we are motivated to ask: Is tunneling \emph{necessary} to achieve a quantum speedup on these problems? In order to answer this question, we demonstrate using the optimal TTS criterion defined in Eq.~\eqref{eqt:TTSopt} that the optimal annealing time for QA is far from adiabatic. Instead, as shown in Fig.~\ref{fig:QA_EnergyOverlap}, the optimal TTS for QA is such that the system leaves the instantaneous ground state for most of the evolution and only returns to the ground state towards the end. The cascade down to the ground state is mediated by a sequence of avoided energy level-crossings as seen in Fig.~\ref{fig:Cascade}.  We consider this a diabatic form of QA (DQA) and call this mechanism through which DQA achieves a speedup a \emph{diabatic cascade}.

 As $n$ increases for fixed $u$, repopulation of the ground state improves for fixed $({t_f})_{\textrm{opt}}$, hence causing TTS$_\mathrm{opt}$ to decrease with $n$, as seen Fig.~\ref{fig:TTSScaling}, until it saturates to a constant at the lowest possible value, corresponding to a single run at $({t_f})_{\textrm{opt}}$. At this point the problem is solved in constant time $({t_f})_{\textrm{opt}}$, compared to the $\sim\mathcal{O}(n^{0.5})$ scaling of the adiabatic regime. Moreover, as shown in Fig.~\ref{fig:QA_O3_AveHW}, there are no sharp changes in $\wich{\mathrm{HW}}$, suggesting that the non-adiabatic dynamics do not entail multi-qubit tunneling events, unlike the adiabatic case. Thus, this establishes that we may have speedups in QA that do not involve multi-qubit tunneling.

One may worry that for this diabatic evolution to be successful, the optimal annealing time may need to be very finely tuned. We address this concern in Appendix~\ref{app:tfprecision}, where we show that if $\epsilon$ is the precision desired in $p_{\mathrm{GS}}$, we need only have a precision of $\mathrm{polylog}(1/\epsilon)$ in setting $t_f$, which means that the diabatic speedup is robust.

Figure~\ref{fig:SpikeMovingVanDam_TTSopt} shows that the speedup of DQA and SVD over AQA exists for three other PHWO problems: the {Moving Plateau}, the {Spike}, and the {0.5-Rectangle} problems. Importantly, DQA and SVD have an exponential speedup over AQA for the 0.5-Rectangle problem. We do not observe a diabatic speedup for the {Precipice} or {Grover} problems.

\begin{figure*}[t] 
   \subfigure[\ Spike]{\includegraphics[width=0.6\columnwidth]{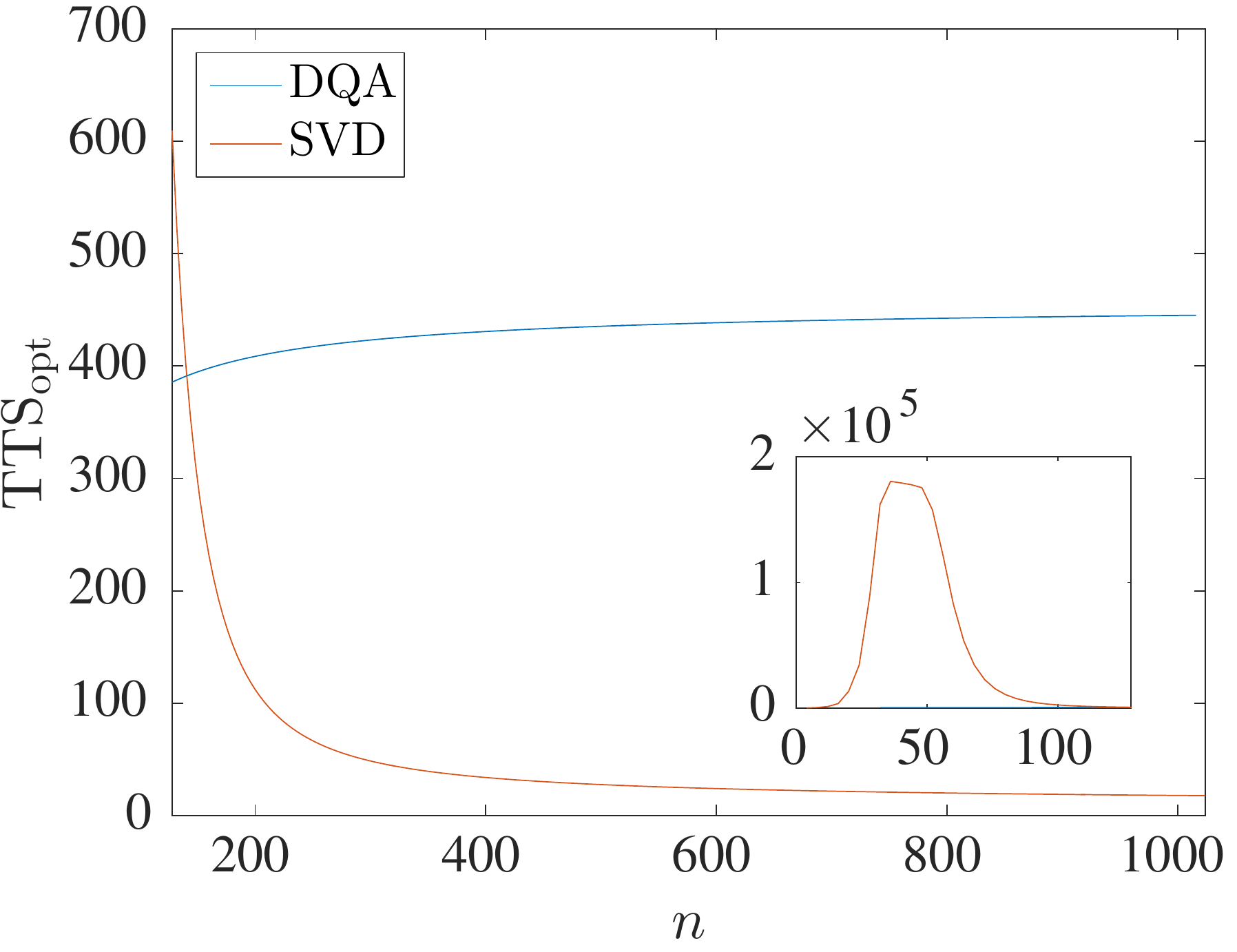}}
    \subfigure[\ Moving Plateau]{\includegraphics[width=0.6\columnwidth]{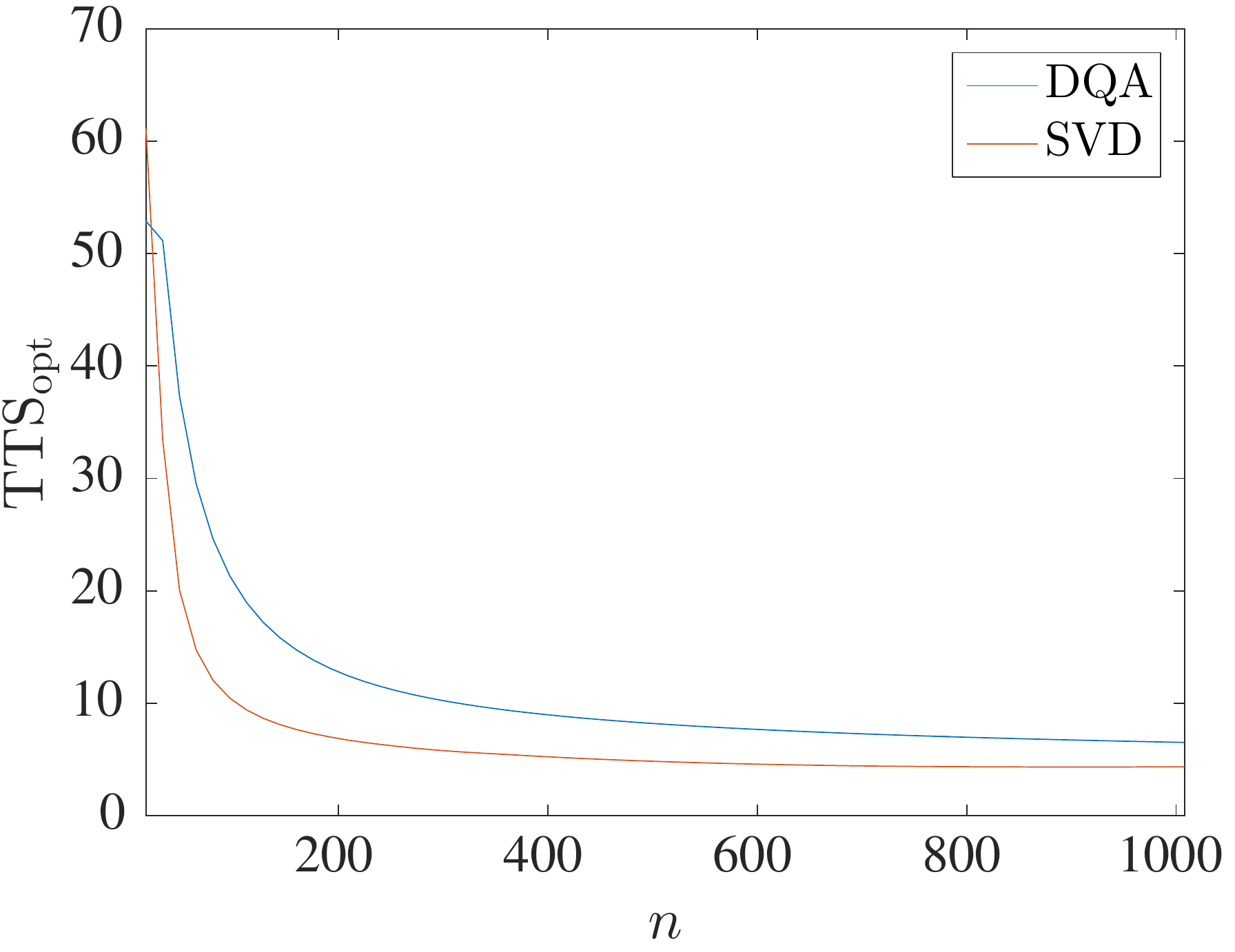}}
   \subfigure[\ 0.5-Rectangle]{\includegraphics[width=0.6\columnwidth]{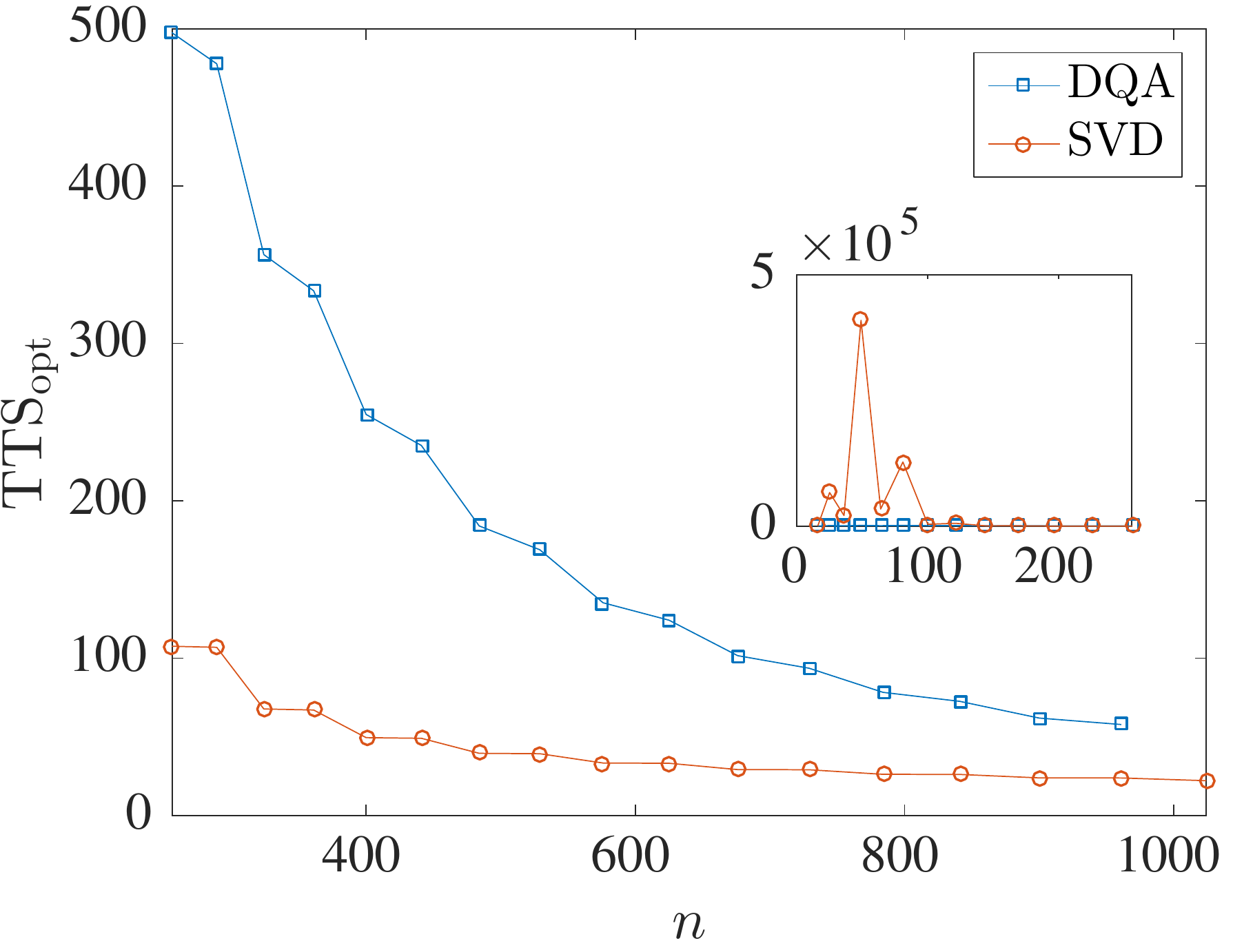}}
   \subfigure[\ Spike]{\includegraphics[width=0.6\columnwidth]{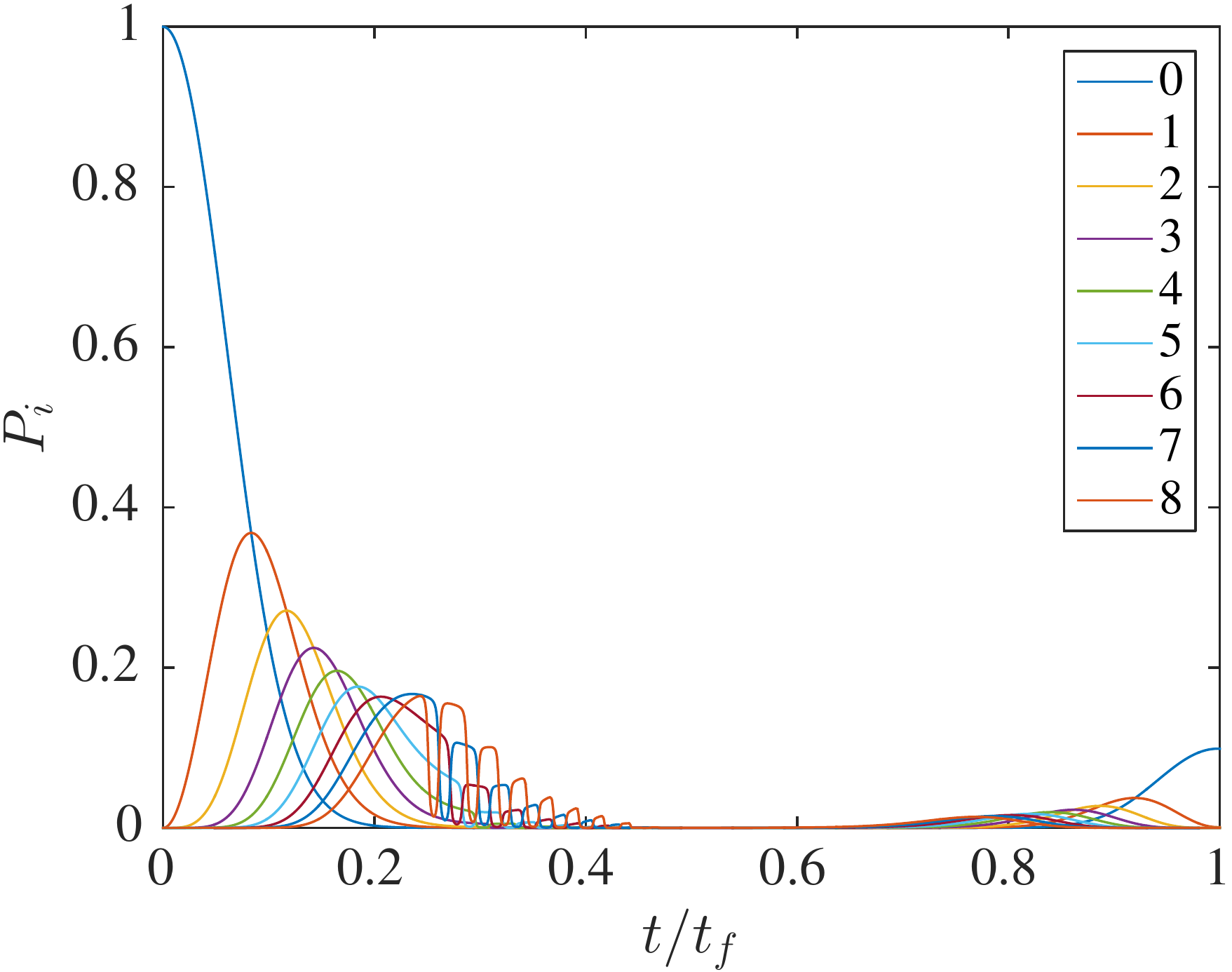}}
   \subfigure[\ Moving Plateau]{\includegraphics[width=0.6\columnwidth]{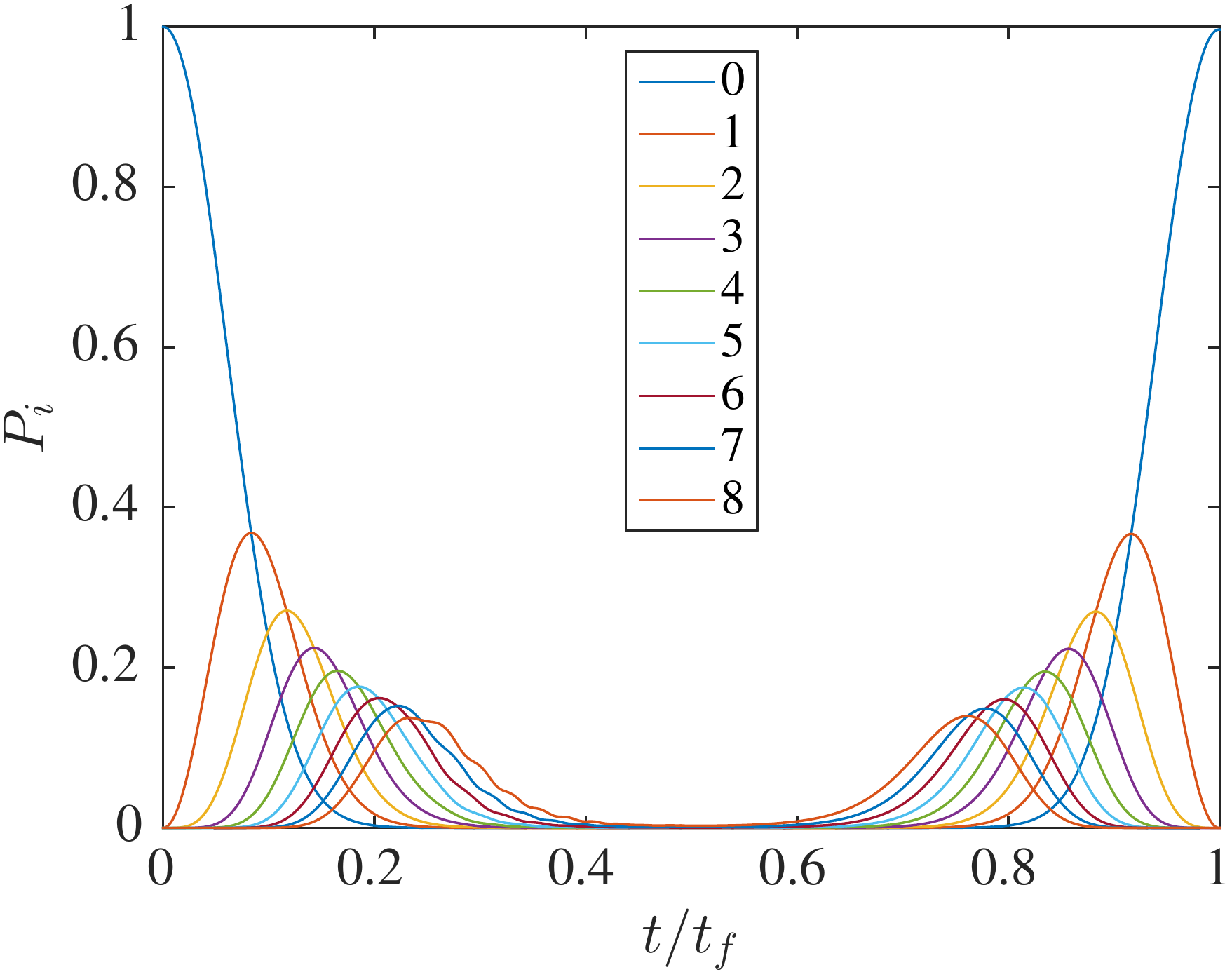}}
   \subfigure[\ 0.5-Rectangle]{\includegraphics[width=0.6\columnwidth]{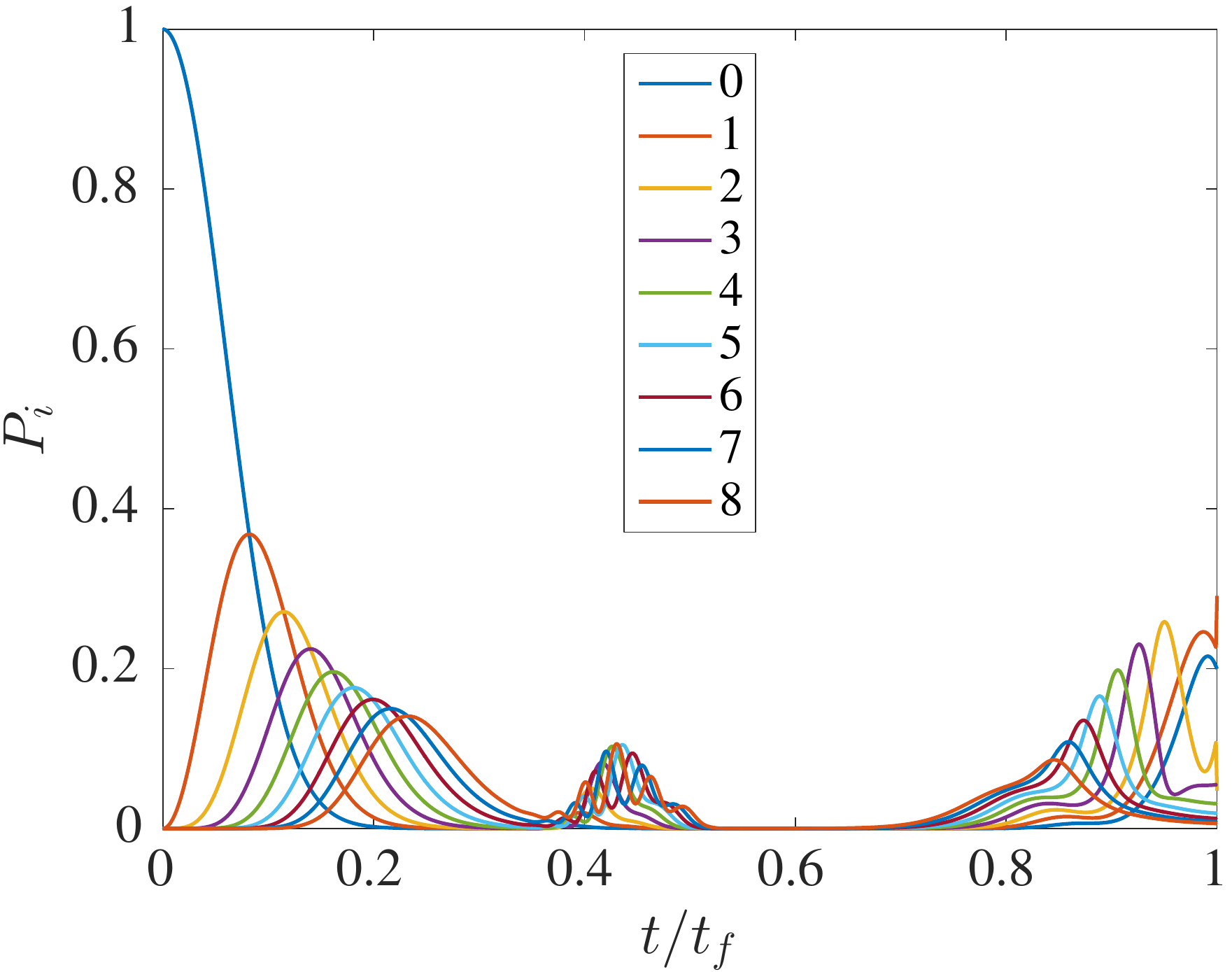}}
   \caption{(a-c) The optimal TTS for the {Spike}, {Moving Plateau}, and {0.5-Rectangle} problems respectively.  Inset for (a) and (c): the optimal TTS for small problem sizes, where we observe SVD at first scaling poorly. However, as $n$ grows, this difficulty vanishes and it quickly outperforms DQA. (d-f) Population $P_i$ in the $i$-th energy eigenstate along the diabatic QA evolution at the optimal TTS. We observe similar diabatic transitions for these problem (shown are the cases with $n=512$ and $t_f = 9.85$ for the Spike, $n=512$ and $t_f = 10$ for the Moving Plateau, and $n=529$ and $t_f = 9.8$ for the 0.5-Rectangle) as we observed for the Fixed Plateau [Fig.~\ref{fig:QA_EnergyOverlap}].}
\label{fig:SpikeMovingVanDam_TTSopt}
\end{figure*}   
%
\subsection{Spin Vector Dynamics} \label{subsec:SVDplat}

Given the absence of tunneling in the time-optimal quantum evolution, we are motivated to consider the behavior of Spin-Vector Dynamics (SVD), which arise in a semi-classical limit (see Appendix~\ref{app:methods_SVD} for an overview of this algorithm).  As we show in Fig.~\ref{fig:TTSScaling}, the scaling of SVD's optimal TTS also saturates to a constant time, i.e., $({t_f})_{\textrm{opt}}$. Moreover, it reaches this value earlier (as a function of problem size $n$) than DQA, thus outperforming DQA for small problem sizes, while for large enough $n$ both achieve $\mathcal{O}(1)$ scaling. As seen in the inset, SVD's advantage persists as a function of $u$ at constant $n$. 

The dynamics of DQA are well approximated by SVD until close to the end of the evolution, as shown in Fig.~\ref{fig:QA_O3_AveHW}: the trace-norm distance between the instantaneous states of DQA and SVD is almost zero until $t/t_f \approx 0.8$, after which the states start to diverge. This suggests that SVD is able to replicate the DQA dynamics up to this point, and only deviates because it is more successful at repopulating the ground state than DQA.

In Fig.~\ref{fig:SpikeMovingVanDam_TTSopt}, we show that SVD's speedup over AQA is replicated for the {Spike}, {Moving Plateau}, and 0.5-Rectangle problems as well.  Remarkably, while the 0.5-Rectangle problem has an exponentially small gap [see Eq.~\eqref{eqt:VanDam} and Fig.~\ref{fig:MinGaps}], SVD and DQA both achieve $\mathcal{O}(1)$ scaling, and hence the diabatic cascades provides an exponential speedup relative to AQA.

It is important to note that SVD is ineffective if one desires to simulate the adiabatic evolution. In the absence of unitary dynamics (which allow for tunneling) or thermal activation (to thermally hop over the barrier), SVD gets trapped behind the barrier that forms in the semi-classical potential separating the two degenerate minima [see Fig.~\ref{fig:Veff}] and is unable to reach the new global minimum. In this sense, SVD does not enjoy the guarantee provided by the quantum adiabatic theorem for the unitary evolution~\cite{Jansen:07,Amin:09,lidar:102106}, that for sufficiently long $t_f$ dictated by the adiabatic condition, the ground state can be reached with any desired probability.

Likewise, it is important to keep in mind the distinction between a classical algorithm being able to match, or sometimes outperform, a quantum algorithm (as SVD does here), and the classical algorithm approximating the evolution or instantiating the physics of the quantum algorithm (as SVD fails to do here). Indeed, in both the diabatic and adiabatic regimes, SVD provides a poor approximation to the instantaneous quantum state. For example, in the diabatic regime, it is clear from Fig.~\ref{fig:QA_O3_AveHW} that the trace-norm distance between the instantaneous SVD state and the instantaneous quantum state starts to increase significantly for $s\gtrsim 0.8$.  In the same spirit, consider the instantaneous semi-classical ground state, i.e., the spin-coherent state evaluated at the minimum of the spin-coherent potential, which may be suspected to provide a good approximation to the instantaneous quantum ground state, but does not as shown in Fig.~\ref{fig:GSdistanceSC-QA}. Thus the unentangled semi-classical ground state also fails to provide a good approximation to the quantum ground state.

\subsection{Simulated Quantum Annealing}\label{subsec:SQAplat}
Simulated Quantum Annealing (SQA) is a quantum Monte Carlo algorithm performed along the annealing schedule (see Appendix~\ref{app:methods_SQA} for further details). It is often used as a benchmark against which QA is compared (though see Ref.~\cite{Heim:2014jf} for caveats).  SQA scales better than SA for the Fixed Plateau problem using the threshold criterion (see Fig.~\ref{fig:ThC}).  In order to understand why SQA enjoys an advantage over SA using this benchmark metric, it is useful to study the behavior of the state of SQA along the annealing schedule.  We show the behavior of $\langle \mathrm{HW} \rangle$ for SQA in Fig.~\ref{fig:AveHWSQAA}, where we observe that SQA at the optimal number of sweeps (the case of $1500$ sweeps shown in Fig.~\ref{fig:AveHWSQAA}) does \emph{not} follow the instantaneous ground state.  Instead it reaches the threshold success probability by thermally relaxing to the ground state \emph{after} the minimum gap point (and tunneling event) of the quantum Hamiltonian.  Therefore, SQA's advantage over SA stems from the fact that it thermalizes in a different energy landscape than SA.

We also contrast the behavior of SQA and AQA using the threshold criterion.  While SQA is able to follow the instantaneous ground state for a sufficiently large number of sweeps and thus mimic the tunneling of AQA (see Fig.~\ref{fig:AveHWSQAA}), this is not the optimal way for it to reach the threshold criterion.  For a fixed threshold success probability, the process of thermal relaxation after the minimum gap point uses fewer sweeps (and hence is more efficient) than following the instantaneous ground state closely throughout the anneal \footnote{This may be an artifact of our implementation of SQA, whereby we only include cluster updates along the imaginary-time direction and not along the spatial (problem) direction. An implementation with space-like cluster updates may allow SQA via its thermal relaxation to mimic the tunneling of AQA more efficiently.  Whether this is the case will be addressed in future work.}. This is in contrast to AQA, where tunneling is the only means for it to reach a high success probability and nevertheless is more efficient than SQA, as seen in Fig.~\ref{fig:ThC}.
  
\begin{figure}[t] 
   \centering
   \includegraphics[width=0.8 \columnwidth]{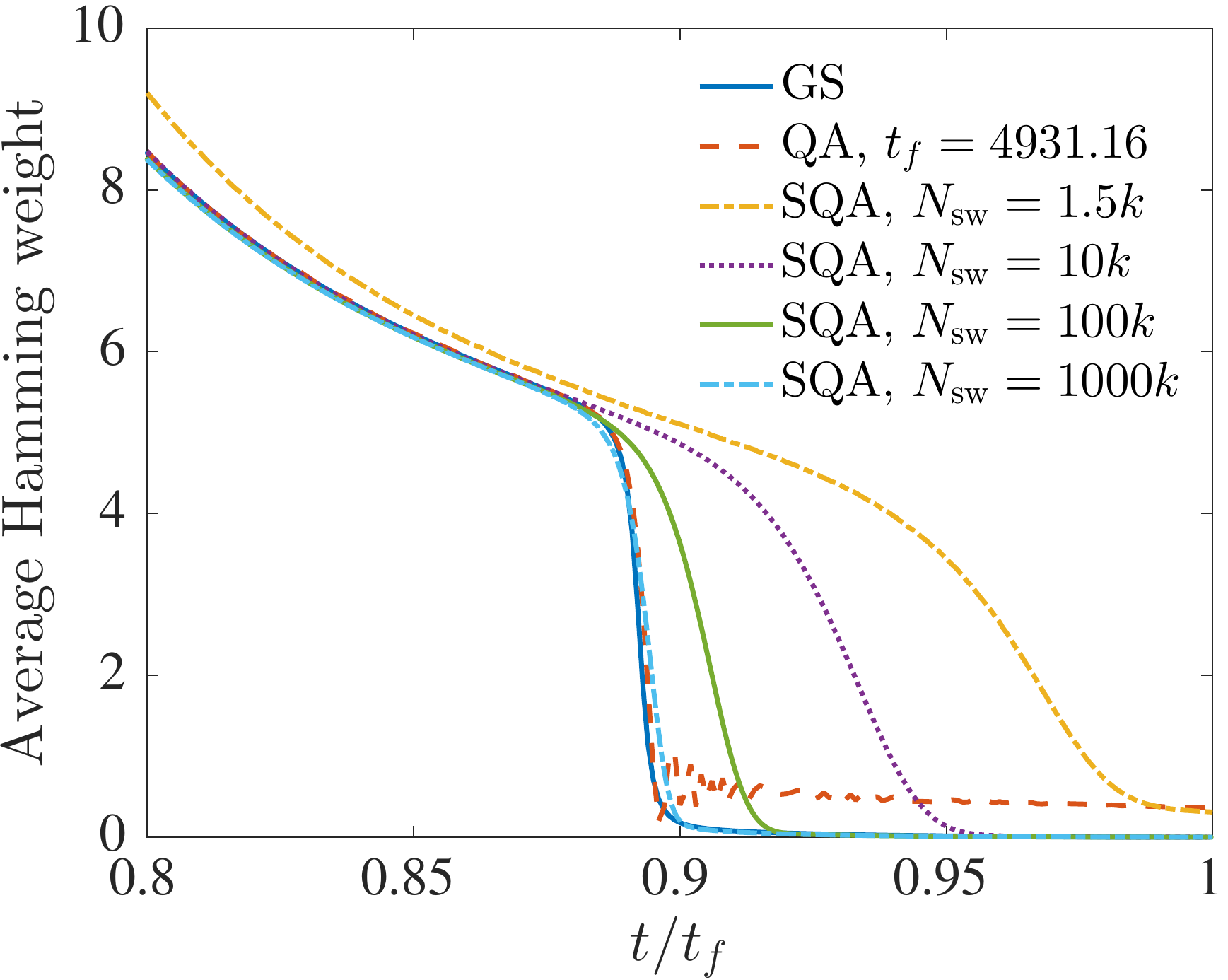} 
   \caption{The expectation value of the Hamming weight operator for the quantum ground state, SQA, and AQA for the {Fixed Plateau} problem with $n=512$, $l = 0$, and $u = 6$ and annealing time chosen so as to reach a success probability of $0.9$.  The expectation value for SQA ($\beta = 30$, $N_{\tau} = 64$) at a given $t/t_f$ is calculated by averaging over the Hamming weight of the $N_\tau$ imaginary time states at that time and over $10^5$ independent trials.  The case of $1500$ sweeps is the minimum number of sweeps required for SQA (in `annealer' mode) to reach the threshold ground state probability of $p_0 =0.9$, and similarly for the annealing time value of $t_f = 4931.16$ for AQA.  While AQA is able to approximately follow the quantum ground state (i.e., the evolution is very close to being adiabatic), the optimal SQA evolution (i.e., that requires the fewest sweeps) for achieving the threshold criterion involves \emph{not} following the ground state at the minimum gap point and instead thermally relaxing towards the ground state after this point.  As shown using the higher $N_{\textrm{sw}}$ values, only after increasing the number of sweeps by more than two orders of magnitude does SQA follow the instantaneous ground state closely.}
   \label{fig:AveHWSQAA}
\end{figure}

We note that SQA's threshold criterion advantage over SA does not carry over to the optimal TTS criterion. In fact, we find that using the optimal TTS criterion, SQA scales as $O(n^{1.5})$, while SA scales as  $O(n)$, as seen in Fig.~\ref{fig:TTSScaling}. The reason for the latter scaling is that the optimal number of sweeps for SA is $1$, simply because there is a small but non-zero probability that in the first sweep all the $1$s are flipped to $0$s.

\section{Discussion} \label{sec:discuss}
%
\begin{figure}[t] 
   \centering
\includegraphics[width=0.32\textwidth]{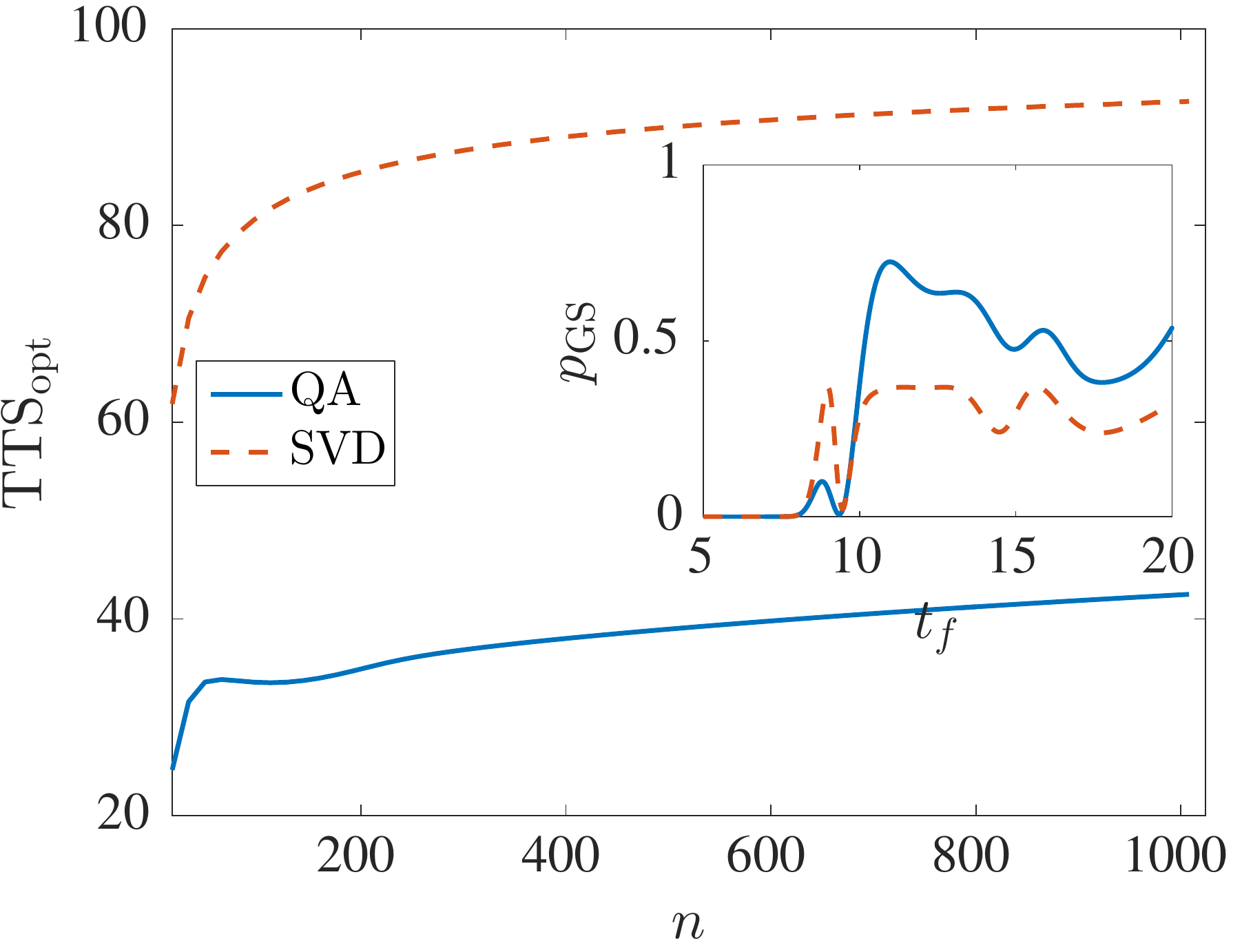}
   \caption{The optimal TTS for the potential given in Eq.~\eqref{eqt:jarret}.  QA outperforms SVD over the range of problem sizes we were able to check. The reason can be seen in the inset, which displays the ground state probability for SVD and QA for different annealing times $t_f$, with $n = 512$.  The optimal annealing time for SVD occurs at the first peak in its ground state probability ($t_f \approx 8.98$), whereas the optimal annealing time for QA occurs at the much larger \emph{second} peak in its ground state probability ($t_f \approx 10.91$).}
   \label{fig:ConvexPGS}
\end{figure}

It is often assumed that the shape of the final cost-function determines how hard it is for QA to solve the problem (in fact, this was partly the motivation for the {Spike} problem in Ref.~\cite{Farhi-spike-problem}), and that potentials with tall and thin barriers should be advantageous for AQA, since this is where tunneling dominates over thermal hopping (e.g., \cite[p.215]{Heim:2014jf}, \cite[p.1062]{RevModPhys.80.1061}, \cite[p.226]{Suzuki-book}). It is then assumed that problems where the final potential has this feature are those for which there should be a quantum speedup. We have given several counterexamples to such claims, and shown that tunneling is not necessary to achieve the optimal TTS. Instead, the optimal trajectory may use diabatic transitions to first scatter completely out of the ground state and return via a sequence of avoided level crossings.  
That diabatic transitions can help speed up quantum algorithms has also been noted and advantageously exploited in Refs.~\cite{Somma:2012kx, crosson2014different, Hen2014,Steiger:2015fk}. Moreover, we have shown that the instantaneous semi-classical potential provides important insight into the role of tunneling, while the final cost function can be rather misleading in this regard. 

While both adiabatic and diabatic QA outperform SA for the Fixed Plateau problem, the faster quantum diabatic algorithm is not better than the classical SVD algorithm for this problem.  The PHWO problems due to Reichardt~\cite{Reichardt:2004}, which includes problems very similar to the Fixed Plateau, have widely been considered an example where tunneling provides a quantum advantage; we have shown that this holds if one limits the comparison to SA, but that there is in fact no quantum speedup in the problem when one compares the quantum diabatic evolution (which outperforms adiabatic quantum annealing) to SVD.

These results of the diabatic optimal evolution extend beyond the plateau problems: 
even the {Spike} problem studied in Ref.~\cite{Farhi-spike-problem}---which is in some sense the antithesis of the plateau problem since it features a sharp spike at a single Hamming weight---also exhibits the diabatic-beats-adiabatic phenomenon, indicating that tunneling is not required to efficiently solve the problem. Thus diabatic evolution, especially via diabatic cascades, is an important and relatively unexplored mechanism in quantum optimization that is different from tunneling. The fact that we observe a speedup relative to AQA for several problems, especially an exponential speedup for the 0.5-Rectangle, motivates the search for algorithms exploiting this mechanism and may yield fruitful results. However, we also already know that diabatic cascades are not generic. E.g., we have checked that this mechanism is absent in the {Grover} and {Precipice} problems, even though the Grover problem is equivalent to a `giant' plateau problem.

In summary, our work provides a counterargument to the widely made claims that tunneling should be understood with respect to the final cost function, that speedups due to tunneling require tall and thin barriers; and that tunneling is needed for a quantum speedup in optimization problems. Which features of Hamiltonians of optimization problems favor diabatic or adiabatic algorithms remains an open question, as is the understanding of tunneling for non-permutation-symmetric problems.

We finish on a positive note for QA. We have given several examples where SVD outperforms QA, e.g., the Spike problem~\cite{Farhi-spike-problem}. However, we make no claim that SVD will always have an advantage over QA. A simple and instructive example comes from the class of cost functions that are convex in Hamming weight space, which have a constant minimum gap \cite{jarret2014fundamental}:
\beq 
\label{eqt:jarret}
f(x) = \begin{cases} 2, &  |x| =0 \\
\abs{x}, & \text{otherwise} \end{cases} \ .
\eeq
We have observed similar diabatic transitions for this problem as for the Fixed Plateau (not shown), but find that DQA outperforms SVD, as shown in Fig.~\ref{fig:ConvexPGS}.  This results because the optimal TTS for QA occurs at a slightly higher optimal annealing time, i.e., there is an advantage to evolving somewhat more slowly, though still far from adiabatically. Thus, this provides an example of a ``limited'' quantum speedup \cite{speedup}.

\acknowledgments{
Special thanks to Ben Reichardt for insightful conversations and for suggesting the plateau problem, and to Bill Kaminsky for inspiring talks~\cite{Kaminsky:2014,Kaminsky:USC-talk-2014}. We also thank Itay Hen, Joshua Job, Iman Marvian, Milad Marvian, and Rolando Somma for useful comments. The computing resources were provided by the USC Center for High Performance Computing and Communications and by the Oak Ridge Leadership Computing Facility at the Oak Ridge National Laboratory, which is supported by the Office of Science of the U.S. Department of Energy under Contract No. DE-AC05-00OR22725.  This work was supported under ARO grant number W911NF-12-1-0523 and ARO MURI Grant No. W911NF-11-1-0268.}

\newpage
\appendix

\section{Review of the Hamming weight problem and Reichardt's bound for PHWO problems} \label{app:review}
Here we closely follow Ref.~\cite{Reichardt:2004}.

\subsection{The Hamming weight problem}

We review the analysis within QA of the minimization of the Hamming weight function 
$f_\mathrm{HW} (x)=\abs{x}$, which counts the number of $1$'s in the bit string $x$. This problem is of course trivial, and the analysis given here is done in preparation for the perturbed problem.

For the adiabatic algorithm, we start with the driver Hamiltonian,
\beq \label{eq:HD}
H_D = \frac{1}{2} \sum_{i=1}^n \left( \ident_i - \sigma^x_i \right) = \sum_{i=1}^n \ket{-}_i\bra{-} \ ,
\eeq
which has $\ket{+}^{\otimes n}$ as the ground state.

The final Hamiltonian for the cost function $f_\mathrm{HW} (x)$ is
\beq
H_P = \frac{1}{2} \sum_{i=1}^n \left( \ident_i - \sigma^z_i \right) = \sum_{i=1}^n \ket{1}_i\bra{1} \ ,
\eeq
which  has $\ket{0}^{\otimes n}$ as the ground state.

We interpolate linearly between $H_D$ and $H_P$:
\begin{align}
H(s) &= (1-s)H_D + sH_P; \quad s\in[0,1] \\
&= \sum_{i=1}^n \frac{1}{2}\begin{pmatrix}  1-s & -(1-s) \\ -(1-s) & 1-s\end{pmatrix}_i + \begin{pmatrix}  0 & 0 \\ 0 & s\end{pmatrix}_i,\\
&= \frac{1}{2}\sum_{i=1}^n \left[\openone -(1-s)\sigma^x_i -s\sigma^z_i\right] \equiv H_i(s) \ .
\label{eqt:PHWSingle}
\end{align}
We note that $H_i(s)$ in Eq.~\eqref{eqt:PHWSingle} is similar to a variant of the Landau-Zener (LZ) Hamiltonian with finite coupling duration \cite{vitanov1996landau,vitanov1996erratum}, for which the Schr\"{o}dinger equation has an analytical solution, except that there it is assumed that the $\sigma^x$ term is constant and only the $\sigma^z$ terms has a (linear) time dependence over a finite interval. The analytical solution of the problem obtained in Ref.~\cite{vitanov1996landau} is rather complicated, and for our purposes a simpler approach suffices.

Since there are no interactions between the qubits, the adiabatic problem can be solved exactly by diagonalizing the Hamiltonian acting on each qubit separately. For each term, we have the energy eigenvalues $E_{\pm}(s)$,
\beq
E_{\pm}(s) = \frac{1}{2} (1 \pm \Delta(s)); \quad \Delta(s) \equiv \sqrt{1-2s+2s^2},
\eeq
and associated eigenvectors,
\beq
\ket{v_{\pm}(s)} = \frac{1}{\sqrt{2\Delta (\Delta \mp s)}} \left[ \mp(\Delta \mp s) \ket{0} + (1-s) \ket{1} \right] \ .
\label{eq:v(s)}
\eeq
The ground state of $H(s)$ is 
\beq
\ket{\psi_\mathrm{GS}(s)}= \ket{v_{-}(s)}^{\otimes n} \ .
\label{eq:psiGS}
\eeq
The gap is given by,
\bes
\begin{align}
\text{Gap}[H(s)] &= H(s) \ket{v_{+}(s)}\otimes\ket{v_{-}(s)}^{\otimes (n-1)} \nonumber \\
& \hspace{3cm} - H(s)\ket{v_{-}(s)}^{\otimes n} \\
&= E_+ + (n-1)E_- - nE_- \\
&= E_+ - E_- \\
&=\Delta(s) \ .
\end{align}
\ees
The gap is minimized at $s=\frac{1}{2}$ with minimum value $\Delta(\frac{1}{2})=\frac{1}{\sqrt{2}}$. The minimum gap is independent of $n$ and hence does not scale with problem size. Therefore we can predict an adiabatic run time to be given by,
\beq
t_f = \mathcal{O}\left( \frac{\|\partial_s H\|}{\Delta^2} \right)= \mathcal{O}(n)\ ,
\eeq
where the $n$-dependence is solely due to $\| \partial_s H\|$ (see Appendix-\ref{app:methods_QD}).  However, this is actually a loose upper bound. We next provide separate numerical and analytical arguments to demonstrate that the actual scaling for AQA is $\mathcal{O}(n^{0.5})$.

\subsubsection{Numerical argument} 
Suppose the adiabatic algorithm runs long enough so as to attain a desired success probability, $p_0$. Let this time be $t_f$. Using the fact that the quantum evolution of the plain Hamming Weight problem is the evolution of $n$ non-interacting qubits, we can express the global ground-state probability in terms of the ground-state probabilities of single qubits. So, if the single qubit ground-state probability for this run-time is $p_{\mathrm{GS}}(t_f)$, then we must have $p_0 = p_{\mathrm{GS}}(t_f)^n$.  

We find numerically (see Fig.~\ref{fig:PlainHWpGS}) that $p_{\mathrm{GS}}(t_f)$ has an envelope that is  excellently approximated by:
\beq
p_{\mathrm{GS}}(t_f) = 1 - \frac{1}{t_f^2} + \mO(t_f^{-3}) \ ,
\eeq
for sufficiently large $t_f$.  We therefore can write:
\beq
\ln p_0 = n \ln p_{\mathrm{GS}}(t_f) \approx n \ln \left( 1  - \frac{1}{t_f^2} \right) \ ,
\eeq
and upon expanding the $\ln$, we extract a tighter scaling for our adiabatic time:
\beq
t_f = \mO(n^{1/2}) \ .
\eeq

\begin{figure}[ht] 
   \centering
   \includegraphics[width=0.32\textwidth]{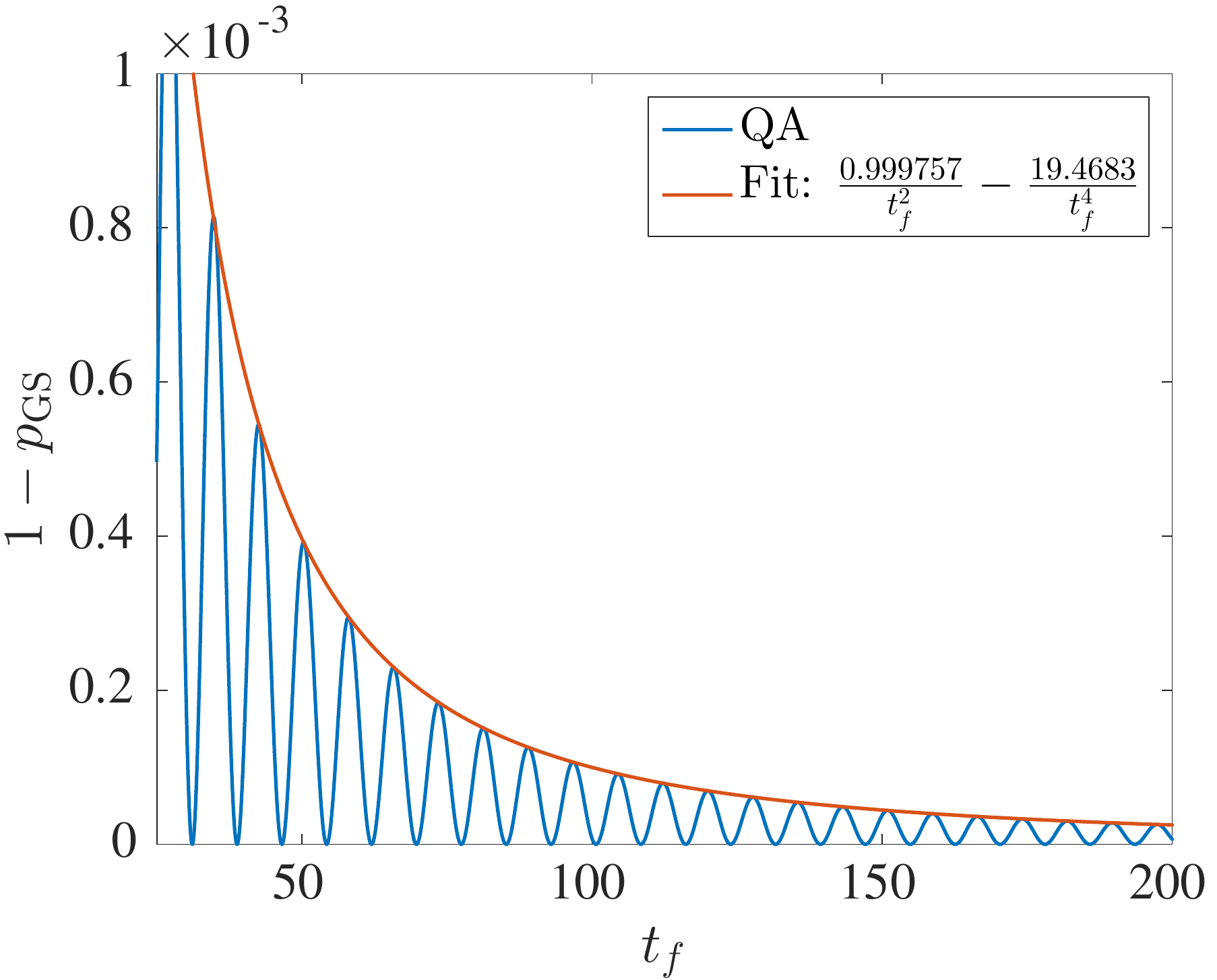} 
   \caption{Ground state probability for a single qubit for different total time $t_f$, evolving under the Plain Hamming Weight Hamiltonian in Eq.~\eqref{eqt:PHWSingle}.}
   \label{fig:PlainHWpGS}
\end{figure}

\subsubsection{Analytical argument} 
Here, we invoke a result due to Boixo and Somma~\cite{Boixo:2010fk}. This result states,
\begin{theorem}[\cite{Boixo:2010fk}] To adiabatically prepare a final eigenstate using a Hamiltonian evolution $H(s)$ requires time that scales at least as $\mathcal{O}\left(\frac{L}{\Delta}\right)$. Here $L$ is the eigenpath length, 
\beq
L \equiv \int_0^1 \| \ket{\partial_s \psi(s)} \| ds,
\eeq
where $\ket{\psi(s)}$ is the eigenpath traversed to reach the final eigenstate.
\end{theorem}

We analytically compute $L$ for the ground-state path in the plain Hamming weight problem, and show that it scales as $\mathcal{O}(\sqrt{n})$. Since we know that in this case $\Delta = \mathcal{O}(1)$, we conclude the adiabatic algorithm will require at least $\mathcal{O}(\sqrt{n})$ time.

Recall that the instantaneous ground state is [Eq.~\eqref{eq:psiGS}]
$\ket{\psi_\mathrm{GS}(s)}= \bigotimes_{i=1}^n \ket{v_-^i(s)}$, where $\ket{v_-^i(s)} =  \sqrt{1-q(s)} \ket{0}_i + \sqrt{q(s)} \ket{1}_i$, with [Eq.~\eqref{eq:v(s)}]
\beq
q(s) = \frac{(1-s)^2}{2\Delta(\Delta+s)}\ .
\label{eq:q(s)}
\eeq
Differentiating:
\beq
\frac{d}{d s} \ket{\psi_\mathrm{GS}(s)} = \sum_{i=1}^n \left(\bigotimes_{j \neq i} \ket{v_-^j(s)} \otimes \frac{d}{d s}\ket{v_-^i(s)}\right) \ ,
\eeq
so that
\begin{align}
&\| \ket{\partial_s \psi_\mathrm{GS}(s))} \|^2 \equiv \wich{\partial_s \psi_\mathrm{GS}(s))|\partial_s \psi_\mathrm{GS}(s))} \\
&= n \| \frac{d}{d s}\ket{v_-^i(s)} \|^2 + n(n-1) |\wich{v_-^i(s)|\frac{d}{d s}|v_-^i(s)}|^2.
\end{align}
The term $\| \frac{d}{d s}\ket{v_-^i(s)} \|$ does not have any scaling with $n$, and the second term vanishes because it is equal to $\frac{1}{2} \frac{d}{d s}\wich{v_-^i(s)|v_-^i(s)}=0$, where we use the fact that $\ket{v_-^i(s)}$ is real-valued and normalized. 
Thus, taking the square root on both sides and integrating from $0$ to $1$, we obtain the $\sqrt{n}$ scaling of $L$. 

If we desire to fix the constant in front of $L$, a straightforward calculation will show that
\beq
\int_0^1 \| \frac{d}{d s}\ket{v_-^i(s)} \| ds = \pi/4 \ .
\eeq
\subsection{Reichardt's bound for PHWO problems} 
\label{app:ReichardtBound}
Here we review Reichardt's derivation of the gap lower-bound for general PHWO  problems, but provide additional details not found in the original proof~\cite{Reichardt:2004}. 

We use the same initial Hamiltonian [Eq.~\eqref{eq:HD}] and linear interpolation schedule as before, $\tilde{H}(s) = (1-s) H_D + s \tilde{H}_P$, and choose the final Hamiltonian to be
\beq 
\tilde{H}_P = \sum_{x\in \{0,1\}^n} \tilde{f} (x) \ket{x}\bra{x}\ ,
\eeq
where
\beq
\tilde{f}(x) = \begin{cases} \abs{x} + p(x) & l<\abs{x}<u \ ,\\
\abs{x} & \text{elsewhere} \end{cases} \ , 
\eeq
where $p(x) \geq 0$ is the perturbation.  Note that here we have \emph{not} assumed that the perturbation, $p(x)$, respects qubit permutation symmetry. 

We wish to bound the minimum gap of $\tilde{H}(s)$. Unlike the Hamming weight problem $H(s)$, this problem is no longer non-interacting. Define
\beq
h_k \equiv \max_{\abs{x}=k} p(x); \quad h \equiv \max_k h_k = \max_x p(x).
\eeq

\begin{lemma}[\cite{Reichardt:2004}]
\label{lem:bound}Let $u = \mathcal{O}(l)$ and let $E_0(s)$ and $\tilde{E}_0(s)$ be the ground state energies of $H(s)$ and $\tilde{H}(s)$, respectively. Then $\tilde{E}_0 (s) \leq E_0(s) + \mathcal{O}( h \frac{u-l}{\sqrt{l}})$.
\end{lemma}

\begin{proof}
First note that
\beq
\tilde{H}(s)- H(s) = s \sum_{x: l<\abs{x}<u} p(x) \ket{x}\bra{x}\ .
\eeq

Below, we suppress the $s$ dependence of all the terms for notational simplicity.  We know that $E_0 = \wich{v_{-}^{\otimes n}|H|v_{-}^{\otimes n}}$. Using this,
\bes
\begin{align}
\wich{\tilde{E}_0|\tilde{H}|\tilde{E}_0} &\leq \wich{\psi | \tilde{H} | \psi} \quad \forall \ket{\psi}\in\mathcal{H}. \\
\implies \tilde{E}_0 - E_0  &\leq \wich{v_{-}^{\otimes n}|\tilde{H}|v_{-}^{\otimes n}} - E_0 \\
&\leq \wich{v_{-}^{\otimes n}|\tilde{H}-H|v_{-}^{\otimes n}} \label{eq:perttheory}\\
&=s \sum_{x: l<\abs{x}<u} p(x) \abs{\wich{v_{-}^{\otimes n}|x}}^2 \\
&=s \sum_{x: l<\abs{x}<u} p(x) q^{\abs{x}}(1-q)^{n-\abs{x}} \\
&\leq \sum_{k:l<k<u} h_k {n \choose k}  q^k (1-q)^{n-k}, \label{eq:hk}
\end{align}
\ees
where ${n \choose k}$ is the number of strings with Hamming weight $k$, we used the fact that if we measure in the computational basis, the probability of getting outcome $x$ is $\abs{\wich{v_{-}^{\otimes n}|x}}^2 = q(s)^{\abs{x}} (1-q(s))^{n-\abs{x}}$, and $q(s)$ is given in Eq.~\eqref{eq:q(s)}.

Consider the partial binomial sum (dropping the $h_k$'s),
\beq
\sum_{k:l<k<u} {n \choose k}  q^k (1-q)^{n-k}.
\eeq
Using the fact that the binomial is well-approximated by the Gaussian in the large $n$ limit (note that this approximation requires that $q(s)$ and $1-q(s)$ not be too close to zero), we can write:
\begin{align}
\label{eq:integralbound}
& \sum_{k:l<k<u} {n \choose k}  q^k (1-q)^{n-k}  \approx \int_l^u d\xi \ \frac{1}{\sqrt{2 \pi}\sigma} e^{-\frac{(\xi - \mu)^2}{2 \sigma^2}} \notag \\
&\quad = \frac{1}{\sigma} \int_l^u d\xi\ \phi\left(\frac{\xi-\mu}{\sigma}\right) = \int_{(l-\mu)/\sigma}^{(u-\mu)/\sigma} dt\ \phi(t)\ , 
\end{align}
where $\quad \mu \equiv nq$, $\sigma \equiv \sqrt{nq(1-q)}$ and $\phi(t) \equiv \frac{e^{-t^2/2}}{\sqrt{2\pi}}$. Note that $\sigma$ and $\mu$ depend on $n$, and also on $s$ via $q(s)$. The parameters $l$ and $u$ are specified by the problem Hamiltonian, and are therefore allowed to depend on $n$ as long as $l(n)<u(n)<n$ is satisfied for all $n$.

Let us define:
\beq
\label{eq:intbound}
B(s,n,l(n),u(n)) \equiv \int_{(l(n)-\mu(n,s))/\sigma(n,s)}^{(u(n)-\mu(n,s))/\sigma(n,s)} dt\ \frac{e^{-t^2}/2}{\sqrt{2\pi}}.
\eeq
We seek an upper bound on this function.  We observe that $q(s)$ decreases monotonically from $\frac{1}{2}$ to $0$ as $s$ goes from $0$ to $1$.   Thus, the mean of the Gaussian $\mu(n,s) = n q(s)$ decreases from $\frac{n}{2}$ to $0$. 
Depending on the values of $l(n)$, $u(n)$ and $\mu(n,s)$, we thus have three possibilities: (i) $l(n)<\mu(n,s)<u(n)$, (ii) $\mu(n,s)<l(n)<u(n)$, and (iii) $l(n)<u(n)<\mu(n,s)$. Note that (ii) and (iii) are cases where the integral runs over the tails of the Gaussian and so the integral is exponentially small. We focus on (i), as this induces the maximum values of the integral.
In this case the lower limit of the integral Eq.~\eqref{eq:intbound} is negative, while the upper limit is positive. Thus, the integral runs through the center of the standard Gaussian, and we can upper-bound the value of the integral by the area of the rectangle of width $\frac{u(n)-l(n)}{\sigma(n,s)}$ and height $\frac{1}{\sqrt{2\pi}}$. Hence
\bes
\begin{align}
B(s,n,l(n),u(n)) &\leq \frac{1}{\sqrt{2\pi}} \frac{u(n)-l(n)}{\sigma(n,s)}, \\
&= \frac{1}{\sqrt{2\pi}} \frac{u(n)-l(n)}{\sqrt{\mu(n) (1-q(s))}}, \\
&\leq \frac{1}{\sqrt{2\pi}} \frac{u(n)-l(n)}{\sqrt{l(n) (1-q(s))}},
\end{align}
\ees
where we have used the fact that $l(n) < \mu(n,s) = nq(s)$.

Thus, we obtain the bound:
\beq
\tilde{E}_0 - E_0 \leq \mathcal{O}\left(h \frac{u-l}{\sqrt{l}}\right). \label{eq:bound}
\eeq
\end{proof}

\begin{lemma}[\cite{Reichardt:2004}] 
\label{lem:spec} 
If $\tilde{H} - H$ is non-negative, then the spectrum of $\tilde{H}$ lies above the spectrum of $H$. That is, $\tilde{E}_j \geq E_j$ for all $j$, where $\tilde{E}_j$ and $E_j$ denote the $j$th largest eigenvalue of $\tilde{H}$ and $H$, respectively.
\end{lemma}

This can be proved by a straightforward application of the Courant-Fischer min-max theorem (see, for example, Ref.~\cite{horn2012matrix}).

Combining these lemmas results in the desired bound on the gap:
\bes
\begin{align}
\text{Gap}[\tilde{H}(s)] &= \tilde{E}_1 - \tilde{E}_0, \\
&\geq E_1 - \tilde{E}_0,  \label{eq:lem1}\\
&= E_1 - E_0 - (\tilde{E}_0 - E_0), \\
&\geq \Delta - \mathcal{O}\left(h \frac{u-l}{\sqrt{l}}\right) \label{eq:lem2},
\end{align}
\ees
where in Eq.~\eqref{eq:lem1} we used Lemma \ref{lem:spec} and in Eq.~\eqref{eq:lem2}, we used Lemma \ref{lem:bound}.

Now, if we choose a parameter regime for the perturbation such that $h \frac{u-l}{\sqrt{l}} = o(1)$, then the perturbed problem maintains a constant gap. For example, if $l = \Theta(n)$ and $h(u-l) = \mathcal{O}(n^{1/2-\epsilon})$, for any $\epsilon > 0$, then the gap is constant as $n \to \infty$.

\section{(Non-)Locality of PHWO problems} \label{app:locality}
%

Since the PHWO problems, including the plateau, are quantum oracle problems, they cannot generically be represented by a local Hamiltonian. For completeness we prove this claim here and also show why the (plain) Hamming weight problem is $1$-local.

Let $r$ be a bit string of length $n$, i.e., $r \in \{0,1\}^n$ and let 
\beq
\sigma^r \equiv \sigma_1^{r_1} \otimes \sigma_2^{r_2} \otimes \dots \otimes \sigma_n^{r_n},
\eeq
with $ \sigma_i^0 \equiv I_i$ and $\sigma_i^1 \equiv \sigma_i^z$.  This forms an orthonormal basis for the vector space of diagonal Hamiltonians. Thus:
\beq
H_P = \sum_{r \in \{0,1\}^n} J_r \sigma^r,
\eeq
with
\bes
\begin{align}
J_r &= \frac{1}{2^n}\Tr(\sigma^r H_P) \\
& = \frac{1}{2^n}\sum_{x \in \{0,1\}^n} f(x) \bra{x}\sigma^r \ket{x} \\
& = \frac{1}{2^n}\sum_{x \in \{0,1\}^n} f(x) (-1)^{x \cdot r}.
\end{align}
\ees
Note that generically $J_r$ will be be non-zero for arbitrary-weight strings $r$, leading to $|r|$-local terms in $H_P$, even as high as $n$-local.

E.g., substituting the plateau Hamiltonian [Eq.~\eqref{eqt:plateau}] into this we obtain:
\begin{eqnarray}
J_r &=& \frac{1}{2^n} \left[ \sum_{\abs{x} \leq l\, \&\, \abs{x} \geq u} \abs{x} (-1)^{x \cdot r} \right. \nonumber \\
&& \left. + (u-1)\sum_{l<\abs{x}<u} (-1)^{x \cdot r} \right].
\end{eqnarray}

On the other hand, if $f(x)=|x|$ (i.e., in the absence of a perturbation), the Hamiltonian is only $1$-local:
\bes
\begin{align}
H_P &= \sum_{x \in \{0,1\}^n} |x| \ket{x}\bra{x} \\
&= \sum_{x_1=0}^1 \dots \sum_{x_n =0}^1 (x_1 + x_2 + \dots + x_n) \ket{x_1}\bra{x_1} \nonumber \\
& \hspace{2cm} \otimes\ket{x_2}\bra{x_2} \otimes \dots \otimes \ket{x_n}\bra{x_n} \\
&= \sum_{k=1}^n \left(x_k \ket{x_k}\bra{x_k} \right) \bigotimes_{j\neq k} \left(\sum_{x_j=0}^1 \ket{x_j}\bra{x_j} \right) \\
&= \sum_{k=1}^n \ket{1}_k \bra{1} \bigotimes_{j\neq k} I_j = \sum_{k=1}^n \ket{1}_k \bra{1}.
\end{align}
\ees
%

\section{Derivation of E\lowercase{q}.~(1)} \label{app:TTS}
Equation~\eqref{eqt:TTSopt} is easily derived as follows: the probability of successively failing $k$ times is $\left[1-p_{\mathrm{GS}}(t_f) \right]^k$, so the probability of succeeding at least once after $k$ runs is $1-\left[1-p_{\mathrm{GS}}(t_f) \right]^k$, which we set equal to the desired success probability $p_d$; from here one extracts the number of runs $k$ and multiplies by $t_f$ to get the time-to-solution TTS. Optimizing over $t_f$ yields TTS$_\mathrm{opt}$, which is natural for benchmarking purposes in the sense that it captures the trade-off between repeating the algorithm many times \textit{vs} optimizing the probability of success in a single run. The adiabatic regime might be more attractive if one seeks a theoretical guarantee to have a certain probability of success if the evolution is sufficiently slow.

\section{Methods} \label{app:Methods}
%

\subsection{Simulated Annealing}
\label{app:SA}
%
SA is a general heuristic solver \cite{kirkpatrick_optimization_1983}, whereby the system is initialized in a high temperature state, i.e., in a random state, and the temperature is slowly lowered while undergoing Monte Carlo dynamics.  Local updates are performed according to the Metropolis rule \cite{1953JChPh..21.1087M,HASTINGS01041970}: a spin is flipped and the change in energy $\Delta E$ associated with the spin flip is calculated.  The flip is accepted with probability $P_{\mathrm{Met}}$:
\beq
P_{\mathrm{Met}} = \min \{ 1 , \exp(-\beta \Delta E) \} \ ,
\eeq
where $\beta$ is the current inverse temperature along the anneal. Note that there could be different schemes governing which spin is to be selected for the update. We consider two such schemes: random spin-selection -- where the next spin to be updated is selected at random; and sequential spin-selection -- where one runs through all of the $n$ spins in a sequence. Random spin-selection (including just updating nearest neighbors) satisfies detailed-balance and thus is guaranteed to converge to the Boltzmann distribution.  Sequential spin-selection does not satisfy strict detailed balance (since the reverse move of sequentially updating in the reverse order never occurs), but it too converges to the Boltzmann distribution \cite{Manousiouthakis:1999}.  In sequential updating, a ``sweep" refers to all the spins having been updated once. In random spin-selection, we define a sweep as the total number of spin updates divided by the total number of spins. When it is possible to parallelize the spin updates, the appropriate metric of time-complexity is the number of sweeps $N_{\mathrm{SW}}$, not the number of spin updates (they differ by a factor of $n$) \cite{speedup}. However, in our problem this parallelization is not possible and hence the appropriate metric is the number of spin updates, and this is what is plotted in Fig.~\ref{fig:TTSScaling}.  After each sweep, the inverse temperature is incremented by an amount $\Delta \beta$ according to an annealing schedule, which we take to be linear, i.e., $\Delta \beta = ( \beta_f - \beta_i ) / (N_{\mathrm{SW}}-1)$.

We can use SA both as an \emph{annealer} and as a \emph{solver} \cite{Hen:2015rt}. In the former, the state at the end of the evolution is the output of the algorithm, and can be thought of as a method to sample from the Boltzmann distribution at a specified temperature. For the latter, we select the state with the lowest energy found along the entire anneal as the output of the algorithm, the better technique if one is only interested in finding the global minimum. We use the latter to maximize the performance of the algorithm.

%
\subsection{Quantum Annealing}\label{app:methods_QD}

Here we consider the most common version of  quantum annealing:
\beq
H(s) = (1-s) \sum_{i=1}^n \frac{1}{2}(\ident_i - \sigma_i^x) + s\sum_{x \in \{0,1\}^n} f(x) \ket{x}\bra{x} \ ,
\eeq
where $s\equiv {t}/{t_f}$ is the dimensionless time parameter and $t_f$ is the total anneal time. The initial state is taken to be $\ket{+}^{\otimes n}$, which is the ground state of $H(0)$.

The initial ground state and the total Hamiltonian are symmetric under qubit permutations (recall that $f(x) = f(|x|)$ for our class of problems).  It then follows that the time-evolved state, at any point in time, will also obey the same symmetry. Therefore the evolution is restricted to the $(n+1)$-dimensional symmetric subspace, a fact that we can take advantage of in our numerical simulations.  This symmetric subspace is spanned by the Dicke states $\ket{S, M}$ with $S = n/2, M = -S, -S+1, \dots ,S$, which satisfy:
\bes
\begin{align}
S^2  \ket{S, M} &= S \left( S+1 \right) \ket{S, M}\\
S^z \ket{S, M} &= M \ket{S, M} \ ,
\end{align}
\ees
where $S^{x,y,z} \equiv \frac{1}{2} \sum_{i=1}^n \sigma_i^{x,y,z}$, $S^2 = \left(S^x \right)^2 + \left(S^y \right)^2 +\left(S^z \right)^2$.
We can denote these states by:
\beq \label{eqt:Dicke}
\ket{w} \equiv \Ket{\frac{n}{2} ,M = \frac{n}{2} - w} = {n \choose w}^{-1/2} \sum_{x: |x| = w} \ket{x}, 
\eeq
where, $w\in\{0,\dots,n\}$.

In this basis the Hamiltonian is tridiagonal, with the following matrix elements:
\bes \begin{align}
\left[H(s)\right]_{w,w+1} = & \left[H(s)\right]_{w+1,w} =\nonumber \\  &-\frac{1}{2} (1-s)  \sqrt{(n-w)(w+1)}, \\
\left[H(s)\right]_{w,w} = &(1-s) \frac{n}{2} + s  f(w).
\end{align} 
\ees
The Schr\"odinger equation with this Hamiltonian can be solved reliably using an adaptive Runge-Kutte Cash-Karp method \cite{Cash:1990:VOR:79505.79507} and the Dormand-Prince method \cite{Dormand198019} (both with orders $4$ and $5$).

If the quantum dynamics is run adiabatically the system remains close to the ground state during the evolution, and an appropriate version of the adiabatic theorem is satisfied. For evolutions with a non-zero spectral gap for all $s \in [0,1]$, an adiabatic condition of the form
\beq 
t_f \geq \text{const} \sup_{s \in [0,1]} \frac{ \|\partial_s H(s) \|}{\text{Gap}(s)^2}
\label{eq:trad-crit}
\eeq
is often claimed to be sufficient \cite{Messiah:vol2} [however, see the discussion after Eq.~(21) in Ref.~\cite{Jansen:07}]. In our case $\|\partial_s H(s) \| = \| H(1) - H(0) \|$ is upper-bounded by $n$;  since we are considering a constant gap, the adiabatic algorithm can scale at most linearly by condition \eqref{eq:trad-crit}. This is true for the plateau problems.

We showed in the main text that the following version of the adiabatic condition, known to hold in the absence of resonant transitions between energy levels \cite{Amin:09}, estimates the scaling we observe very well:
\beq
\max_{s \in [0,1]}  \frac{ |\bra{\eps_0(s) } \partial_s H(s) \ket{\eps_1(s)}|}{\text{Gap}(s)^2} \ll t_f,
\eeq
where $\eps_0(s)$ and $\eps_1(s)$ are the instantaneous ground and excited states in the symmetric subspace respectively. 

The permutation symmetry is explicitly enforced only in our 
numerical simulations of the quantum evolution. Since, of course, we do not have quantum hardware that can implement the problems under consideration, we must explicitly enforce this symmetry in order to be able to perform numerical simulations at large problem sizes. Note that even if we were to simulate the quantum system without explicitly imposing this symmetry, the symmetry would be automatically preserved in the dynamics, and we would draw the same lessons about the performance of the quantum algorithm (but our classical simulations would quickly become intractable). 

\subsection{Spin-Vector Dynamics }\label{app:methods_SVD}
%
Starting with the spin-coherent path integral formulation of the quantum dynamics, we can obtain Spin Vector Dynamics (SVD) as the saddle-point approximation (see, for example, Ref.~\cite[p.10]{owerre2015macroscopic} or Refs.~\cite{Smolin,Albash:2014if}). It can be interpreted as a semi-classical limit describing coherent single qubits interacting incoherently.  In this sense, SVD is a well motivated classical limit of the quantum evolution of QA.  SVD describes the evolution of $n$ unit-norm classical vectors under the Lagrangian (in units of $\hbar = 1$):
\beq
\mathcal{L} = i \wich{\Omega(s) | \frac{d}{ds} | \Omega(s)} - t_f \wich{\Omega(s) | H(s) | \Omega(s)},
\eeq
where $\ket{\Omega(s)}$ is a tensor product of $n$ independent spin-coherent states~\cite{arecchi1972atomic}:
\beq
\ket{\Omega(s)} = \bigotimes_{i=1}^n \left[ \cos\left( \frac{\theta_i(s)}{2} \right) \ket{0}_i + \sin \left( \frac{\theta_i(s)}{2} \right) e^{i \varphi_i(s)} \ket{1}_i \right].
\eeq
We can define an \emph{effective semi-classical potential} associated with this Lagrangian:
\begin{align}
\label{eqt:bigvsc}
&V_{\mathrm{SC}}(\{\theta_i\},\{\varphi_i\},s) \equiv \wich{\Omega(s) | H(s) | \Omega(s)} \nonumber \\
&=(1-s) \sum_{i=1}^n \frac{1}{2} \left(1- \cos\varphi_i(s) \sin\theta_i(s) \right) \\
&+ s \sum_{x \in \{0,1\}^n} f(x) \prod_{j:x_j=0} \cos^2\left(\frac{\theta_j(s)}{2}\right)\prod_{j:x_j=1} \sin^2\left(\frac{\theta_j(s)}{2}\right),\nonumber
\end{align}
with the probability of finding the all-zero state at the end of the evolution (which is the ground state in our case), as $\prod_{i=1}^n \cos^2\left(\frac{\theta_i(1)}{2}\right)$.
The quantum Hamiltonian obeys qubit permutation symmetry: $P H P = H$ where $P$ is a unitary operator that performs an arbitrary permutation of the qubits. This implies that our classical Lagrangian obeys the same symmetry:
\begin{eqnarray}
\mathcal{L}' &\equiv&  i \bra{\Omega(s)} P \frac{d}{ds} P \ket{\Omega(s)} - t_f \bra{\Omega(s)} P H(s) P \ket{\Omega(s)} \nonumber \\
&=&   i \bra{\Omega(s)}  \frac{d}{ds} \ket{\Omega(s)} - t_f \bra{\Omega(s)}  H(s) \ket{\Omega(s)} = \mathcal{L},  \nonumber\\
\end{eqnarray}
where the derivative operator is trivially permutation-symmetric. Therefore, the Euler-Lagrange equations of motion derived from this action will be identical for all spins. Thus, if we have symmetric initial conditions, i.e., $(\theta_i(0),\varphi_i(0)) = (\theta_j(0),\varphi_j(0)) \ \forall i,j$, then the time evolved state will also be symmetric:
\beq
(\theta_i(s),\varphi_i(s)) = (\theta_j(s),\varphi_j(s)) \ \forall i,j \ \forall s\in[0,1] \ .
\eeq 
As we show below, under the assumption of a permutation-symmetric initial condition we only need to solve two (instead of $2n$) semi-classical equations of motion:
\bes 
\label{eqt:symsvdeoms} 
\begin{align}
\frac{n}{2} \sin\theta(s) \theta'(s) - t_f  \partial_{\varphi(s)} V_{\mathrm{SC}}^\mathrm{sym} (\theta(s),\varphi(s),s)  = 0  \ , \\
- \frac{n}{2} \sin \theta(s) \varphi'(s) - t_f  \partial_{\theta(s)} V_{\mathrm{SC}}^\mathrm{sym} (\theta(s),\varphi(s),s) = 0 \ ,
\end{align} \ees
where we have defined the symmetric effective potential $V_{\mathrm{SC}}^\mathrm{sym}$ as:
\begin{align}
\label{eqt:symvsc}
&V_{\mathrm{SC}}^\mathrm{sym} (\theta(s),\varphi(s),s) \equiv \wich{\Omega^\mathrm{sym}(s) | H(s) | \Omega^\mathrm{sym}(s)} \nonumber \\
&=(1-s) \frac{n}{2} \left(1- \cos\varphi(s) \sin\theta(s) \right) \nonumber \\
 &+ s \sum_{w=0}^n f(w) \binom{n}{w} \sin^{2w}\left(\frac{\theta(s)}{2}\right) \cos^{2(n-w)}\left(\frac{\theta(s)}{2}\right),
\end{align}
and $\ket{\Omega^\mathrm{sym}(s)}$ is simply $\ket{\Omega(s)}$ with all the $\theta$'s and $\varphi$'s set equal. Note that in the main text [see Eq.~\eqref{eq:vsc}], we slightly abuse notation for simplicity, and use $V_{\mathrm{SC}}$ instead of $V_{\mathrm{SC}}^\mathrm{sym}$.
The probability of finding the all-zero bit string at the end of the evolution is accordingly given by $\cos^{2n}(\theta(1)/2)$. We would have arrived at the same equations of motion had we used the symmetric spin coherent state in our path integral derivation, but that would have been an artificial restriction. In our present derivation the symmetry of the dynamics naturally imposes this restriction.

Note that the object in Eq.~\eqref{eqt:bigvsc} involves a sum over all $2^n$ bit-strings and is thus exponentially hard to compute; on the other hand, the object in Eq.~\eqref{eqt:symvsc} only involves a sum over $n$ terms and is thus easy to compute. Therefore, just as in the quantum case---where due to permutation symmetry the quantum evolution is restricted to the $n+1$ dimensional subspace of symmetric states instead of the full $2^n$-dimensional Hilbert space---given knowledge of the symmetry of the problem we can efficiently compute the SVD potential and efficiently solve the SVD equations of motion. 

We also remark that the computation of the potential in Eq.~\eqref{eqt:bigvsc} is significantly simplified if our cost function, $f(x)$, is given in terms of a local Hamiltonian. For example, if $H(1) = \sum_{i,j} J_{ij} \sigma_i^z \sigma_j^z$, then:
\beq
V_{\mathrm{SC}}(\{\theta_i\},\{\varphi_i\},1) = \sum_{i,j} J_{ij} \cos\theta_i \cos\theta_j \ ,
\eeq
which is easy to compute as it is a sum over $\mathrm{poly}(n)$ number of terms.

Let us now derive the symmetric SVD equations of motion~\eqref{eqt:symsvdeoms}. Without any restriction to symmetric spin-coherent states, the SVD equations of motion, for the pair $\theta_i,\varphi_i$, read:
\bes 
\label{eqt:symsvdeoms2} 
\begin{align}
\frac{1}{2} \sin\theta_i(s) \theta_i'(s) - t_f  \partial_{\varphi_i(s)} V_{\mathrm{SC}} (\{\theta_i\},\{\varphi_i\},s)  = 0 \ , \\
- \frac{1}{2} \sin \theta_i(s) \varphi_i'(s) - t_f  \partial_{\theta_i(s)} V_{\mathrm{SC}}(\{\theta_i\},\{\varphi_i\},s) = 0 \ .
\end{align} 
\ees
As can be seen by comparing Eqs.~\eqref{eqt:symsvdeoms} and \eqref{eqt:symsvdeoms2}, it is sufficient to show that: 
\beq
\frac{\partial}{\partial \theta_i} V_{\mathrm{SC}} \bigg|_{\theta_j = \theta, \varphi_j = \varphi \ \forall j} = \frac{1}{n} \frac{\partial}{\partial \theta} V_{\mathrm{SC}}^\mathrm{sym},
\eeq
and an analogous statement holding for derivatives with respect to $\varphi$. This claim is easily seen to hold true for the term multiplying $(1-s)$ in Eq.~\eqref{eqt:bigvsc}:
\begin{align}
& \frac{\partial}{\partial \theta_i} \sum_{i=1}^n \frac{1}{2} \left(1- \cos\varphi_i(s) \sin\theta_i(s) \right)\bigg|_{\theta_j = \theta, \varphi_j = \varphi \ \forall j} \notag \\
& = \frac{\partial}{\partial \theta} \frac{1}{2} \left(1- \cos\varphi(s) \sin\theta(s) \right) \notag \\
& = \frac{1}{n} \frac{\partial}{\partial \theta}  V_{\mathrm{SC}}^\mathrm{sym} (\theta,\phi,s=0)\ ,
\end{align}
where in the last line we used Eq.~\eqref{eqt:symvsc}. Next we focus on the term multiplying $s$ in Eq.~\eqref{eqt:bigvsc}. This term has no $\varphi$ dependence and thus we only consider the $\theta$ derivatives. First note that
\begin{align}
&\frac{\partial}{\partial \theta_i} V_{\mathrm{SC}}(\{\theta_i\}, \{\varphi_i\},s=1) = \nonumber \\
&\sum_{x\in \{0,1\}^n } f(x) \prod_{j:x_j=0} \cos^2\left(\frac{\theta_j}{2}\right)\prod_{j:x_j=1} \sin^2\left(\frac{\theta_j}{2}\right) \nonumber \\
&\times \left[ -\delta_{x_i,0} \sec^2\left(\frac{\theta_i}{2}\right) + \delta_{x_i,1}\csc^2\left(\frac{\theta_i}{2}\right) \right] \frac{\sin\theta_i}{2}.
\end{align}
Now, we set all the $\theta_i$'s equal. Let us define $p(\theta) \equiv \sin^2\left(\frac{\theta}{2}\right)$. Using this and the fact that $f$ is only a function of the Hamming weight (which is equivalent to the qubit permutation symmetry), we can rewrite the last expression, after a few steps of algebra, as:
\begin{align}
\label{eq:32}
&\sum_{w=0}^n f(w) p^{w-1} (1-p)^{n-w-1} \partial_\theta p \nonumber \\
& \hspace{3cm} \times \left[ (1-p)\binom{n-1}{w-1} - p \binom{n-1}{w} \right] \nonumber \\
&= \sum_{w=0}^n f(w) p^{w-1} (1-p)^{n-w-1} \partial_\theta p  \left[ \frac{1}{n} \binom{n}{w} (w - np) \right] \notag \\
& = \frac{1}{n} \frac{\partial}{\partial \theta} V_{\mathrm{SC}}^\mathrm{sym}(\theta, \varphi, s=1)\ .
\end{align}

Similar to the quantum case, we can perform SVD without explicitly imposing the permutation symmetry, and obtain the same results. Here too, we are forced to explicitly exploit the symmetry due to the non-local nature of the problem under consideration, which makes directly implementing the SVD oracle (without the symmetry) exponentially hard. For local problems we can efficiently implement the SVD oracle.

In the results presented in the main text, it is the implementation of SA that does not share this symmetry.  However, while the quantum algorithms and SVD can be implemented without knowledge of the symmetry and still retain their advantage, an implementation of SA that uses the symmetry would require intimate knowledge of the problem. This would be an unfair advantage for SA, not for the quantum evolution.

\subsection{Simulated Quantum Annealing}\label{app:methods_SQA}
An alternative method to simulated annealing, simulated quantum annealing (SQA, or Path Integral Monte Carlo along the Quantum Annealing schedule) \cite{sqa1,Santoro} is an annealing algorithm based on discrete-time path-integral quantum Monte Carlo simulations of the transverse field Ising model using Monte Carlo dynamics.  At a given time $t$ along the anneal, the Monte Carlo dynamics samples from the Gibbs distribution defined by the action:
\beq
S[\mu] =  \Delta(t) \sum_\tau H_{\mathrm{P}} (\mu_{:,\tau}) -J_{\perp}(t)  \sum_{i,\tau} \mu_{i,\tau} \mu_{i,\tau+1} 
\eeq
where $\Delta(t) = \beta B(t)/ N_{\tau}$ is the spacing along the time-like direction, $J_{\perp} = - \ln [\tanh(A(t)/2)]/2$ is the ferromagnetic spin-spin coupling along the time-like direction, and $\mu$ denotes a spin configuration with a space-like direction (the original problem direction, indexed by $i$) and a time-like direction (indexed by $\tau$).
For our spin updates, we perform Wolff cluster updates \cite{PhysRevLett.62.361} along the imaginary-time direction only.  For each space-like slice, a random spin along the time-like direction is picked. The neighbors of this spin are added to the cluster (assuming they are
parallel) with probability 
\beq
P = 1 - \exp(-2 J_{\perp})
\eeq
When all neighbors of the spin have been checked, the newly added spins are checked.  When all spins in the cluster have had their neighbors along the time-like direction tested, the cluster is flipped according to the Metropolis probability using the space-like change in energy associated with flipping the cluster.  A single sweep involves attempting to update a single cluster on each space-like slice.

As in SA, we can use SQA both as an \emph{annealer} and as a \emph{solver} \cite{Hen:2015rt}.  In the former, we randomly pick one of the states on the Trotter slices at the end of the evolution as the output of the algorithm, while for the latter, we pick the state with the lowest energy found along the entire anneal as the output of the algorithm.  We use the latter to maximize the performance of the algorithm.

\section{Behavior of $p_\mathrm{GS}$ \emph{vs.} $t_f$ curves} \label{app:tfprecision}

\begin{figure*}[t] 
   \subfigure[]{\includegraphics[width=0.4\textwidth]{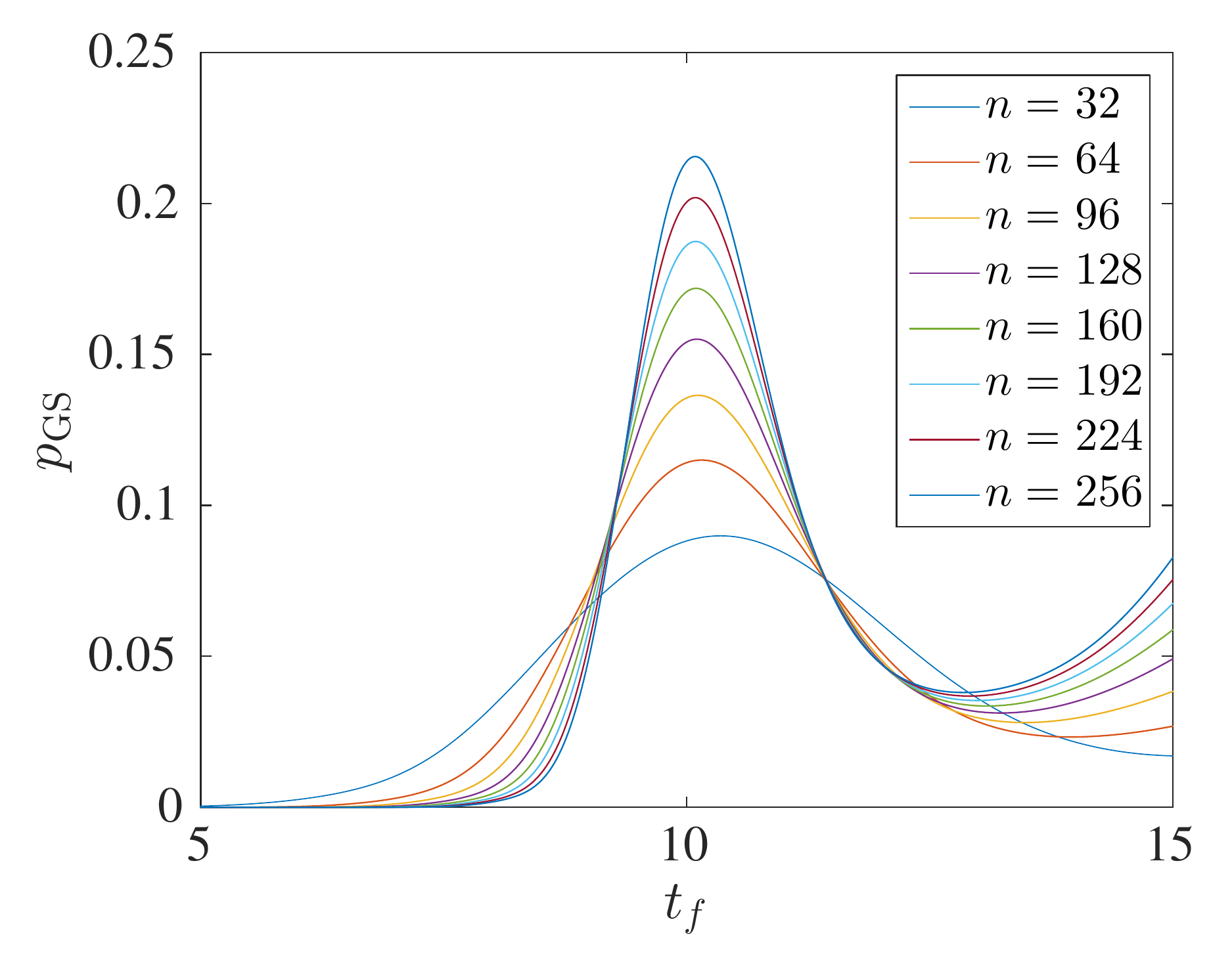} \label{fig:pgsvstfcurves}}
   \subfigure[]{\includegraphics[width=0.4\textwidth]{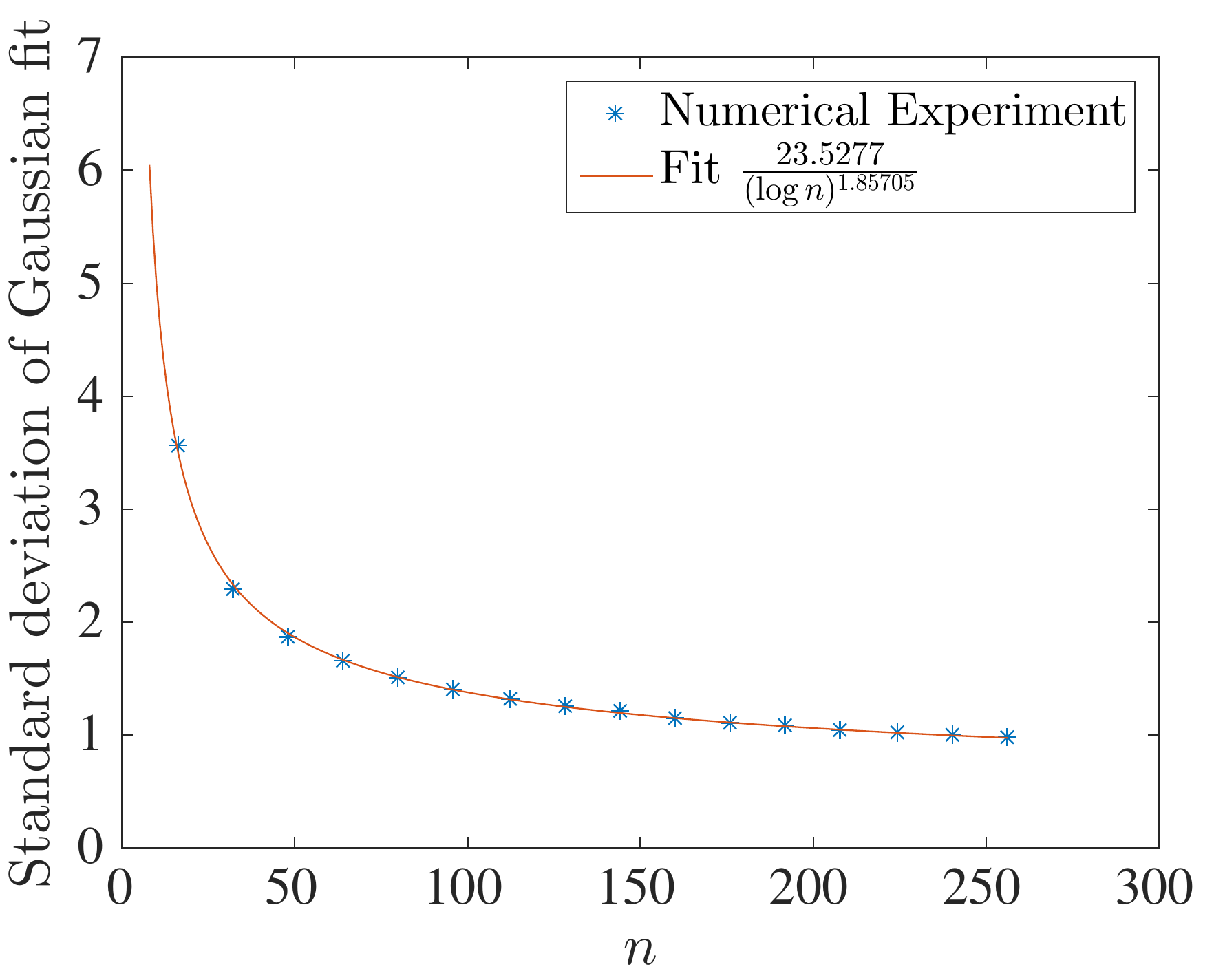}\label{fig:gaussianfitscaling}}
   \caption{Scaling of the $p_\mathrm{GS}$ peak of QA for the Fixed Plateau with $l=0$ and $u=6$. (a) $p_\mathrm{GS}$ \emph{vs.} $t_f$ curves for several problem sizes $n$. The peak at $t_f \approx 10$ is the cause of the optimal annealing time. (b) For each $p_\mathrm{GS}$ \emph{vs.} $t_f$ curve, we fit a Gaussian to the peak. Plotted is the standard deviation of the fitted Gaussians as a function of $n$. An inverse polylogarithmic function provides an excellent fit to this data.}
   \label{fig:widthscaling}
\end{figure*}   

We found that for many of the PHWO problems studied, the optimal $t_f$ lies around $t_f = 10$. This is because there is a peak in the probability of finding the ground state, $p_\mathrm{GS}$ at this $t_f$. Moreover, we found that this peak becomes increasingly higher as the problem size, $n$, grows. This is what allows the problem to have an $\mathcal{O}(1)$ scaling. Since this peak becomes increasingly sharper with growing $n$, there may be the worry that one might need an arbitrarily high precision in setting $t_f \approx t_f^\mathrm{opt}$. We address this concern by showing that in fact the width of the $p_\mathrm{GS}$ \emph{vs.} $t_f$ curve decreases as $\mathcal{O}[1/ \mathrm{polylog}(n)]$ for the Fixed Plateau. This shows that we only require a polylogarithmically increasing precision in our ability to set $t_f$ at the optimal value in order to obtain the speedup. 

The evidence is summarized in Fig.~\ref{fig:widthscaling}. The first plot, \ref{fig:pgsvstfcurves}, shows $p_\mathrm{GS}$ \emph{vs.} $t_f$ curves for several values of $n$. The second plot, \ref{fig:gaussianfitscaling}, shows the scaling of the standard deviation of Gaussian fit to the peak at $t_f^\mathrm{opt}$. This scaling is well matched by polylogarithmic fit.

%


\end{document}